\documentclass[preprint, aip, jmp]{revtex4-1}
\usepackage{amsmath}
\usepackage{amssymb,amsthm}
\usepackage{graphicx}
\usepackage{enumerate}
\usepackage{natbib}
\usepackage{mathrsfs} 
\usepackage[all]{xy}

\newtheorem*{thm*}{Theorem}
\newtheorem{prop}{Proposition}
\newtheorem{lemma}{Lemma}
\newtheorem{cor}{Corollary}
\theoremstyle{definition}
\newtheorem{definition}{Definition}

\providecommand{\norm}[1]{\lVert#1\rVert}
\providecommand{\abs}[1]{\lvert#1\rvert}
\providecommand{\inner}[2]{\langle#1,#2\rangle}

\begin{document}

\title{Classical Limits of Unbounded Quantities by Strict Quantization}
\author{Thomas L. Browning}
\affiliation{Department of Mathematics, University of California, Berkeley, CA, USA}
\author{Benjamin H.~Feintzeig}
\affiliation{Department of Philosophy, University of Washington, Seattle, WA, USA}
\author{Robin Gates-Redburg}
\affiliation{Department of Mathematics, Western Washington University, Bellingham, WA, USA}
\author{Jonah Librande}
\affiliation{Department of Mathematics, University of Washington, Seattle, WA, USA}
\author{Rory Soiffer}
\affiliation{Department of Computer Science, University of Waterloo, Waterloo, Ontario, Canada}
\date{\today}

\begin{abstract}
    This paper extends the tools of C*-algebraic strict quantization toward analyzing the classical limits of unbounded quantities in quantum theories.  We introduce the approach first in the simple case of finite systems.  Then we apply this approach to analyze the classical limits of unbounded quantities in bosonic quantum field theories with particular attention to number operators and Hamiltonians.  The methods take classical limits in a representation-independent manner and so allow one to compare quantities appearing in inequivalent Fock space representations.
\end{abstract}

\maketitle

\section{Introduction}

Powerful C*-algebraic tools have been developed in the last few decades for analyzing the classical limits of quantum theories.  These tools form the theory of \emph{strict quantization} \citep{Ri89,Ri93,Ri94,La98b,La06,La17} (in contrast to \emph{formal quantization},\citep{Wa05,Wa15} which is used in perturbative quantum field theory).  Working with C*-algebras provides a rich structure in which to construct quantum theories and their classical limits as well as provide physical interpretations.  However, it is sometimes said that one cannot use C*-algebras to model all of the systems of physical interest.  In quantum field theory, for example, researchers often employ more flexible types of *-algebras. \citep{Re16,HoWa01,HoWa10}  One reason is that C*-algebras do not allow one to capture unbounded quantities such as the field operators, field momentum operators, or number operators associated with such systems.  Yet recent developments in the theory of unbounded operator algebras have made precise the relationships between algebras of bounded and unbounded quantities.\citep[e.g.,][]{FrInKu10}  Certain algebras of unbounded quantities can be understood as \emph{completions}, in a relevant topology, of C*-algebras.  In this paper, we leverage this fact to make some first steps toward understanding the classical limits of unbounded quantities starting from the framework of strict quantization.  Our central contribution is to develop tools for taking the classical limits of unbounded quantities, and to illustrate these tools in free bosonic quantum field theories by analyzing classical limits of number operators.

Others have analyzed unbounded quantities in C*-algebraic terms by working in specified Hilbert space representations.  This allows one to consider unbounded operators affiliated with a represented C*-algebra.  This technique is useful for many purposes, but its dependence on a Hilbert space representation has some drawbacks for systems whose kinematical C*-algebras have unitarily inequivalent representations, including quantum field theories.\cite{Ha92}  By contrast, the methods we develop in this paper work directly at the level of the abstract algebras, thus providing a framework one can use to simultaneously compare even unbounded quantities that appear in inequivalent representations.  

The plan of the paper is as follows.  In \S\ref{sec:strict}, we provide background on strict quantization.  In \S\ref{sec:finite}, we develop tools for taking the classical limits of unbounded quantities in a system with finitely many degrees of freedom---e.g., an $n$-particle system.  We use this simpler example to outline key features of a strict quantization that allow one to extend it to unbounded quantities.  This serves as a jumping off point for the generalization of these methods in \S\ref{sec:inf} to linear bosonic field theories with infinitely many degrees of freedom.  In \S\ref{sec:KG} we focus specifically on the Klein-Gordon field and establish the classical limits of number operators associated with inequivalent representations of the kinematical C*-algebra.  In \S\ref{sec:em}, we apply these methods to the free Maxwell field.  We conclude in \S\ref{sec:con} with some discussion.

Much previous work on the classical limit precedes the approach of the present paper.  For example, 
\citet{He74} provides a semi-classical analysis of bosonic field theories in a particular Fock space representation.  Recent work of \citet{AmBrNi19} provides a different framework from the current investigation that also allows one to discuss coherent states in bosonic field theories, which are known\citep{La98b,CoRo12} to be closely related to the Berezin quantization map employed in the current paper.  And \citet{Fa18} extends that previous work in a representation-independent manner, although the Wigner measures used in that work are associated with the alternative to Berezin quantization called Weyl quantization; since Weyl quantization is not positive, the corresponding Wigner measures are in general not positive.\citep[][]{La98b}  We will not provide a comprehensive comparison of different approaches to the classical limit here.  We simply note that our work falls in the tradition of strict quantization, which \citet{Fe20} argues provides an appropriate physical interpretation of based on uniform approximations of observables.  We will also aim to obtain uniform approximations in the classical limit here, even though we rely on pointwise approximations to generate unbounded observables.

\section{Strict Quantization and the Weyl algebra}
\label{sec:strict}

A strict quantization provides the mathematical framework for analyzing classical limits of states and quantities within a C*-algebraic setting by giving one a notion of the limit of a family of C*-algebras.  For background and examples, see \citet{Ri89,Ri93,Ri94} and \citeauthor{La93a}.\citep{La93a,La93b,La98a,La98b,La06,La13,La17}

\begin{definition}
A \emph{strict quantization} consists in a family of C*-algebras $\{\mathfrak{A}_\hbar\}_{\hbar\in[0,1]}$ and a family of quantization maps $\{\mathcal{Q}_\hbar: \mathcal{P}\to \mathfrak{A}_\hbar\}_{\hbar\in[0,1]}$, each of whose domain $\mathcal{P}\subseteq\mathfrak{A}_0$ is a Poisson algebra, with $\mathcal{Q}_0$ the embedding map.  We require further that $\mathcal{Q}_\hbar[\mathcal{P}]$ is norm dense in $\mathfrak{A}_\hbar$ for each $\hbar\in[0,1]$, and that the following conditions are satisfied for all $A,B\in\mathcal{P}$:
\begin{enumerate}[(i)]
\item (\emph{Dirac's condition}) $\lim_{\hbar\to 0}\norm{\frac{i}{\hbar}[\mathcal{Q}_\hbar(A),\mathcal{Q}_\hbar(B)] - \mathcal{Q}_\hbar(\{A,B\})}_\hbar = 0$;

\item (\emph{von Neumann's condition}) $\lim_{\hbar\to 0}\norm{\mathcal{Q}_\hbar(A)\mathcal{Q}_\hbar(B) - \mathcal{Q}_\hbar(AB)}_\hbar = 0$;

\item (\emph{Rieffel's condition}) the map $\hbar\mapsto\norm{\mathcal{Q}_\hbar(A)}_\hbar$ is continuous.
\end{enumerate}
\end{definition}

A strict quantization determines the structure of a continuous field of C*-algebras in which all sections of the form $[\hbar\mapsto\mathcal{Q}_\hbar(A)]$ for $A\in\mathcal{P}$ are continuous (See Thm. II.1.2.4, p. 111 in \citet{La98b}). Thus, a strict quantization defines classical limits of quantities and states as follows.  The classical limit of a family of quantities $\{\mathcal{Q}_\hbar(A)\}_{\hbar\in [0,1]}$ is understood to be the classical quantity $A\in\mathcal{P}$.  A family of states $\{\omega_\hbar\in\mathcal{S}(\mathfrak{A}_\hbar)\}_{\hbar\in[0,1]}$ is called a \emph{continuous field of states} when the map $\hbar\mapsto\omega_\hbar(\mathcal{Q}_\hbar(A))$ is continuous for each $A\in\mathcal{P}$.  The classical limit of a continuous field of states is understood to be the classical state $\omega_0$.

An illustrative example of strict quantization is the quantization of the \emph{Weyl algebra}, which we review now and use later on.  First, we define the Weyl algebra itself.  Let $E$ be a vector space of test functions with a symplectic form $\sigma$.  In \S\ref{sec:finite}, we will focus on the case where $E=\mathbb{R}^{2n}$ and $\sigma$ is the standard symplectic form, but in \S\ref{sec:inf} we will deal with the case where $E$ is infinite-dimensional, so we proceed here with some generality.  The Weyl algebra $\mathcal{W}(E,\hbar\sigma)$ is generated by elements $W_\hbar(F)$ for each $F,G\in E$ with
\begin{align*}
    W_\hbar(F)W_\hbar(G) := e^{-\frac{i\hbar}{2}\sigma(F,G)}W_\hbar(F+G) &&
    W_\hbar(F)^* := W_\hbar(-F).
\end{align*}
The elements of the form $W_\hbar(F)$ are linearly independent, and we denote the their linear span as $\Delta(E,\hbar\sigma)$.  There is a unique C*-norm on $\Delta(E,\hbar\sigma)$ called the minimal regular norm.\citep{MaSiTeVe74,BiHoRi04a} We define $\mathcal{W}(E,\hbar\sigma)$ as the completion of $\Delta(E,\hbar\sigma)$ with respect to this norm.

The commutative algebra $\mathcal{W}(E,0)$ is *-isomorphic to the algebra $AP(E')$ of $\sigma(E',E)$-continuous almost periodic functions on some topological dual $E'$ to $E$, when $E$ is given a vector space topology (See p. 2902 of \citet{BiHoRi04a}).  Thus, the algebra of classical quantities can be interpreted in a natural way as an algebra of functions on a phase space $E'$.  Furthermore, the *-algebra $\Delta(E,0)$ carries a Poisson bracket defined as the extension of
\[\{W_0(F),W_0(G)\} := \sigma(F,G)W_0(F+G)\]
(See p. 334 of \citet{BiHoRi04b} or Eq. 2.15, p. 11 of \citet{HoRiSc08}) and so can serve as the domain of a quantization map.  We note that in the special case where $E = \mathbb{R}^{2n}$, the phase space is the dual $E' = \mathbb{R}^{2n}$ and so $\mathcal{W}(\mathbb{R}^{2n},0)\cong AP(\mathbb{R}^{2n})$.

When $E = \mathbb{R}^{2n}$, we have further information about the algebras $\mathcal{W}(\mathbb{R}^{2n},\hbar\sigma)$ for $\hbar>0$.  These algebras have a familiar Hilbert space representation on the Hilbert space $L^2(\mathbb{R}^n)$, which we denote $\pi^S_\hbar:\mathcal{W}(E,\hbar\sigma)\to\mathcal{B}(L^2(\mathbb{R}^n))$ and call the \emph{Schr\"odinger representation}:
\[(\pi^S_\hbar(W_\hbar(a,b))\psi)(x) := e^{\frac{i\hbar a\cdot b}{2}}e^{ib\cdot x}\psi(x + \hbar a)\]
for all $a,b\in\mathbb{R}^{n}$.  Since the families $t\mapsto \pi^S_\hbar(W_\hbar(ta,tb))$ are weak operator continuous, Stone's theorem (See p. 264 of \citet{ReSi80}) implies that these one-parameter unitary groups have self-adjoint generators.  These generators are unbounded operators, corresponding to the standard position and momentum operators for $n$ particles, and so  this Hilbert space representation reproduces the ordinary formulation of quantum mechanics.

To define a strict quantization, we work with the C*-algebras $\mathfrak{A}_\hbar:= \mathcal{W}(E,\hbar\sigma)$ for each $\hbar\in[0,1]$, where $\mathfrak{A}_0 = \mathcal{W}(E,0)$ contains $\mathcal{P}:=\Delta(E,0)$ as a dense Poisson subalgebra.  We define the \emph{Weyl quantization} maps $\mathcal{Q}^W_\hbar:\Delta(E,0)\to\mathcal{W}(E,\hbar\sigma)$ as the linear extension of
\[\mathcal{Q}^W_\hbar(W_0(F)) := W_\hbar(F).\]
\citet{BiHoRi04b} show that this structure indeed forms a strict quantization.  Thus, this structure allows one to analyze classical limits of states and quantities in the Weyl algebra.

\section{Finite Systems}
\label{sec:finite}

If one knows that quantization maps $\mathcal{Q}_\hbar: \mathcal{P}\to\mathfrak{A}_\hbar$ not only satisfy the conditions (i)-(iii) of a strict quantization, but furthermore are continuous in a locally convex topology, then one can continuously extend these maps to the completions of the respective algebras in that topology.  Recent work on algebras of unbounded operators \citep{BaFrInTr06,FrInKu07,BaFrInTr08,BaFrInTr10,FrInKu10} shows that such a completion of a C*-algebra $\mathfrak{A}$, which we will in general denote by $\tilde{\mathfrak{A}}$, will be at least a partial *-algebra containing unbounded operators with some discernible structure.

For what follows, we will not need the details of the rich structure theory that has been developed for algebras of unbounded operators. \citep{Sc90,In98,AnInTr02}  Instead, the issue we encounter in applying these ideas in quantization is that quantization maps may fail to be continuous, and continuity is required to guarantee a unique extension of a quantization map to a completion.  For example, the Weyl quantization maps defined in the previous section fail to be continuous in the norm, and hence weak, topologies.  This implies that one cannot continuously extend the Weyl quantization maps to the completion of the Weyl algebra.  This is unfortunate because the Weyl quantization maps have some nice properties; they can be defined with the minimal algebraic structure of the Weyl algebra even on an infinite dimensional phase space.  However, another quantization prescription called \emph{Berezin quantization} is known to be continuous in the norm, and hence weak, topologies.  Berezin quantization is well-defined for systems with finitely many degrees of freedom with phase space $\mathbb{R}^{2n}$, but the standard definition involves phase space integrals that are not in general well-defined when $E$ is infinite-dimensional.  Our goal in this section is thus to put Berezin quantization into a minimal algebraic form so that it can be applied even to systems whose phase space is infinite-dimensional.  As we proceed, we will use the simplified example of a system with finitely many degrees of freedom to illustrate the basic concepts of our approach to dealing with classical limits of unbounded operators.

\subsection{Positive Quantization}

We begin by defining the Berezin quantization maps for a system with phase space $\mathbb{R}^{2n}$.  We will work with the algebras $\mathfrak{A}_\hbar := \mathcal{K}(L^2(\mathbb{R}^n))$ of compact operators on $L^2(\mathbb{R}^n)$ for each $\hbar\in(0,1]$ and the algebra $\mathfrak{A}_0:= C_0(\mathbb{R}^{2n})$ of continuous functions vanishing at infinity on $\mathbb{R}^{2n}$, the latter of which contains the dense Poisson subalgebra $\mathcal{P}:= C_c^\infty(\mathbb{R}^{2n})$ of smooth, compactly supported functions.  The Berezin quantization maps involve integrals over phase space of certain functions of coherent states, but since these integrals are not in general meaningful on infinite dimensional phase spaces, we will seek to put the quantization maps in a different form.  A \emph{coherent state} for $(p,q)\in\mathbb{R}^{2n}$ is a vector $\psi_\hbar^{(p,q)}\in L^2({\mathbb{R}^n})$ of the form
\[\psi_\hbar^{(p,q)}(x) := \frac{1}{(\pi\hbar)^{n/4}}\exp\bigg(-\frac{ip\cdot q}{2\hbar} + \frac{ip\cdot x}{\hbar} - \frac{(x-q)^2}{2\hbar}\bigg).\]
Here and in what follows $x^2$ denotes the dot product $x\cdot x$ for any $x\in\mathbb{R}^{m}$.  The Berezin quantization maps $\mathcal{Q}_\hbar^B: C_c^\infty(\mathbb{R}^{2n})\to \mathcal{K}(L^2(\mathbb{R}^n))$ are then defined for each $f\in C_c^\infty(\mathbb{R}^{2n})$ by
\[(\mathcal{Q}_\hbar^B(f)\psi)(x) := \frac{1}{(2\pi\hbar)^n}\int_{\mathbb{R}^{2n}}f(p,q)\psi_\hbar^{(p,q)}(x)\Big\langle\psi_\hbar^{(p,q)},\psi\Big\rangle dpdq\]
for each $\psi\in L^2(\mathbb{R}^n)$, where $\inner{\cdot}{\cdot}$ is the $L^2$ inner product.

It is known that this quantization map is positive, which implies that it is continuous in the norm (See Prop. 1.3.7, p. 47 of \citet{La98b}), and hence weak, topologies.  First, this entails that $\mathcal{Q}^B_\hbar$ extends continuously to a map $C_0(\mathbb{R}^{2n})\to\mathcal{K}(L^2(\mathbb{R}^n))$, which we will denote by the same symbol.  Second, this implies that $\mathcal{Q}_\hbar^B$ extends continuously to the weak completions of the domain and range, which we now denote $\tilde{\mathcal{Q}}_\hbar^B: \tilde{C_0}(\mathbb{R}^{2n})\to\tilde{\mathcal{K}}(L^2(\mathbb{R}^n))$.  The algebra $\tilde{C}_0(\mathbb{R}^{2n})$ contains many unbounded and even discontinuous functions (See Ex. 4.1, p. 371 of \citet{BaFrInTr06}), but for our purposes we note that it at least contains as a subalgebra the algebra of all continuous functions $C(\mathbb{R}^{2n})\subseteq \tilde{C}_0(\mathbb{R}^{2n})$, and so contains unbounded functions $\Phi_0(a,b): \mathbb{R}^{2n}\to \mathbb{C}$ for fixed $a,b\in\mathbb{R}^n$ of the form
\[\Phi_0(a,b)(p,q) := a\cdot p + b\cdot q\]
for each $p,q\in\mathbb{R}^n$, where $\cdot$ denotes the usual dot product.  These functions include standard classical position and momentum observables.  Similarly, $\tilde{\mathcal{K}}(L^2(\mathbb{R}^n))$ contains many unbounded operators (See Ex. 4.3, p. 372 of \citet{BaFrInTr06}), including all operators $\Phi_\hbar(a,b)$ on $L^2(\mathbb{R}^n)$ for fixed $a,b\in\mathbb{R}^n$ of the form
\[(\Phi_\hbar(a,b)\psi)(x) := (i\hbar a\cdot\nabla\psi)(x) + (b\cdot x)\psi(x)\]
acting on the dense domain of vectors $\psi\in C_c^\infty(\mathbb{R}^n)\subset L^2(\mathbb{R}^n)$.  Again, these operators include standard quantum position and momentum observables.

Note that $AP(\mathbb{R}^{2n})\subseteq C(\mathbb{R}^{2n})\subseteq \tilde{C}_0(\mathbb{R}^{2n})$ and so one can directly compare the maps $\tilde{\mathcal{Q}}_\hbar^B$ and $\mathcal{Q}_\hbar^W$ on the domain $\Delta(E,0)\subseteq AP(\mathbb{R}^{2n}$) (where we freely identify $\mathcal{W}(\mathbb{R}^{2n},0)$ with $AP(\mathbb{R}^{2n})$).  The comparison actually follows directly from the known relationship of Weyl quantization with Berezin quantization (See Eq. 2.117, p. 144 of \citet{La98b}) or from the representation of Berezin quantization in terms of Toeplitz operators on a Segal-Bargmann space (See p. 294 of \citet{BeCo86}).  Here, we will establish the comparison in the Schr\"odinger representation of the Weyl algebra on $L^2(\mathbb{R}^n)$ by direct computation. (See also \citet{Wa09} for a generalization related to Rieffel's deformation.)

Define $c_\hbar(a,b) = e^{-\frac{\hbar}{4}(a^2+b^2)}$ for all $(a,b)\in\mathbb{R}^{2n}$.  Let $c\mathcal{Q}^W_\hbar: \Delta(E,0)\to\mathcal{W}(\mathbb{R}^{2n},\hbar\sigma)$ be the linear extension of the map defined on the generators $W_0(a,b)\in AP(\mathbb{R}^{2n})$ by
\[c\mathcal{Q}^W_\hbar(W_0(a,b)) := c_\hbar(a,b)\mathcal{Q}^W_\hbar(W_0(a,b)).\]
The following proposition establishes that $c\mathcal{Q}_\hbar^W$ is equivalent to the extension of $\mathcal{Q}_\hbar^B$.

\begin{prop}
\label{prop:schr}
For any $f\in AP(\mathbb{R}^{2n})$, $\pi^S_\hbar\big(c\mathcal{Q}_\hbar^W(f)\big) = \tilde{\mathcal{Q}}_\hbar^B(f)$.  In other words, the diagram in Fig. \ref{fig:1} commutes.

\begin{figure}[h]
 \centerline{
\xymatrix{
{\Delta(E,0)} \ar[drr]_{\tilde{\mathcal{Q}}^B_\hbar} \ar[rr]^{c\mathcal{Q}^W_\hbar} & & {\mathcal{W}(\mathbb{R}^{2n},\hbar\sigma)} \ar[d]^{\pi^S_\hbar}\\
 & & {\mathcal{B}(L^2(\mathbb{R}^n))}
}
}
  \caption{Commutative diagram for Prop. \ref{prop:schr}.}
  \label{fig:1}
\end{figure}
\end{prop}

\begin{proof}
It suffices to show the identity holds on the generators $W_0(a,b)\in \Delta(E,0)$ for arbitrary $a,b\in\mathbb{R}^n$.

To show this, we first note the Fourier inversion theorem implies that for any $\psi\in L^2(\mathbb{R}^{n})$,
\begin{align*}
\psi(x+\hbar a)\exp\bigg(\frac{i\hbar a\cdot b}{2}-\frac{\hbar a^2}{4}\bigg) = \int_{\mathbb{R}^n}\int_{\mathbb{R}^n}\exp(2\pi i(\hbar a - y)\cdot \xi) \psi(x+y)\exp\bigg(\frac{ib\cdot y}{2}- \frac{y^2}{4\hbar}\bigg) dyd\xi
\end{align*}
Setting $p = 2\pi \hbar\xi$ gives the equation
\begin{align*}
    \psi(x+\hbar a)\exp\bigg(\frac{i\hbar a\cdot b}{2}-\frac{\hbar a^2}{4}\bigg) &= \frac{1}{(2\pi\hbar)^n}\int_{\mathbb{R}^n}\int_{\mathbb{R}^n}\exp\bigg(\frac{i}{\hbar}(\hbar a - y)\cdot p\bigg)\psi(x+y)\exp\bigg(\frac{ib\cdot y}{2} - \frac{y^2}{4\hbar}\bigg)dydp\\
    &=\frac{1}{(2\pi\hbar)^n}\int_{\mathbb{R}^n}\int_{\mathbb{R}^n} \psi(x+y)\exp\bigg(ia\cdot p + \frac{ib\cdot y}{2} - \frac{ip\cdot y}{\hbar} - \frac{y^2}{4\hbar}\bigg)dydp.
\end{align*}
This implies that
\begin{equation*}
\begin{split}
    &(\tilde{\mathcal{Q}}_\hbar^B(W_0(a,b))\psi)(x)\\
    &= \frac{1}{(\pi\hbar)^{n/2}(2\pi\hbar)^n}\int_{\mathbb{R}^n}\int_{\mathbb{R}^n}\int_{\mathbb{R}^n}\psi(y)\exp\bigg(ia\cdot p + ib\cdot q + \frac{ip\cdot x}{\hbar} - \frac{(x-q)^2}{2\hbar} - \frac{ip\cdot y}{\hbar} - \frac{(y-q)^2}{2\hbar}\bigg)dqdydp\\
    &\begin{split} = \frac{1}{(\pi\hbar)^{n/2}(2\pi\hbar)^n}\int_{\mathbb{R}^n}\int_{\mathbb{R}^n}\int_{\mathbb{R}^n}\psi(y)\exp\bigg(-\frac{q^2}{\hbar} + \big(ib + &\frac{x}{\hbar} + \frac{y}{\hbar}\big)\cdot q + ia\cdot p\\ &+ \frac{ip\cdot x}{\hbar} - \frac{x^2}{2\hbar} - \frac{ip\cdot y}{\hbar}- \frac{y^2}{2\hbar}\bigg)dqdydp\end{split}\\
    &= \frac{1}{(2\pi\hbar)^n}\int_{\mathbb{R}^n}\int_{\mathbb{R}^n}\psi(y)\exp\bigg(\frac{\hbar}{4}\big(ib + \frac{x}{\hbar} +\frac{y}{\hbar}\big)^2 + ia\cdot p + \frac{ip\cdot x}{\hbar} - \frac{x^2}{2\hbar} - \frac{ip\cdot y}{\hbar}- \frac{y^2}{2\hbar}\bigg)dydp\\
    &\begin{split}= \frac{1}{(2\pi\hbar)^n}\int_{\mathbb{R}^n}\int_{\mathbb{R}^n}\psi(y)\exp\bigg(-\frac{\hbar b^2}{4} + \frac{x^2}{4\hbar} + \frac{y^2}{4\hbar} + \frac{ib\cdot x}{2} + &\frac{ib\cdot y}{2} + \frac{x\cdot y}{2\hbar} + ia\cdot p\\
    &+ \frac{ip\cdot x}{\hbar} - \frac{x^2}{2\hbar} - \frac{ip\cdot y}{\hbar}- \frac{y^2}{2\hbar}\bigg)dydp\end{split}\\
    &= \frac{1}{(2\pi\hbar)^n}\int_{\mathbb{R}^n}\int_{\mathbb{R}^n}\psi(y)\exp\bigg(-\frac{\hbar b^2}{4} + \frac{ib\cdot x}{2} + \frac{ib\cdot y}{2} + \frac{x\cdot y}{2\hbar} + ia\cdot p + \frac{ip\cdot x}{\hbar} - \frac{x^2}{4\hbar} - \frac{ip\cdot y}{\hbar}- \frac{y^2}{4\hbar}\bigg)dydp\\
    &= \frac{1}{(2\pi\hbar)^n}\int_{\mathbb{R}^n}\int_{\mathbb{R}^n}\psi(y)\exp\bigg(-\frac{\hbar b^2}{4} + ib\cdot x + ia\cdot p + \frac{ib\cdot (y-x)}{2} - \frac{ip\cdot (y-x)}{\hbar} - \frac{(y-x)^2}{4\hbar}\bigg)dydp\\
    &= \exp\bigg(-\frac{\hbar b^2}{4} + ib\cdot x\bigg)\frac{1}{(2\pi\hbar)^n}\int_{\mathbb{R}^n}\int_{\mathbb{R}^n}\psi(x+y)\exp\bigg(ia\cdot p + \frac{ib\cdot y}{2} - \frac{ip\cdot y}{\hbar} - \frac{y^2}{4\hbar}\bigg)dydp\\
    &= \exp\bigg(-\frac{\hbar b^2}{4} + ib\cdot x\bigg)\exp\bigg(\frac{i\hbar a\cdot b}{2} - \frac{\hbar a^2}{4}\bigg)\psi(x+\hbar a)\\
    &= c(a,b)\exp\bigg(\frac{i\hbar a\cdot b}{2} + ib\cdot x\bigg)\psi(x+\hbar a)\\
    &= (\pi^S_\hbar(c\mathcal{Q}^W_\hbar(W_0(a,b)))\psi)(x),
    \end{split}
\end{equation*}
which is the desired result.
\end{proof}

\noindent The scalars $c_\hbar$ form what \citet{HoRi05} call \emph{quantization factors}, satisfying:
\begin{enumerate}[(a)]
\item $c_\hbar(a,b) \in\mathbb{R}_+$ for all $\hbar\in[0,1]$ and $a,b\in\mathbb{R}^n$;
\item $c_\hbar(0,0) = e^0 = 1$ and $c_0(a,b) = e^0 = 1$ for all $\hbar\in[0,1]$ and $a,b\in\mathbb{R}^n$; and
\item $\hbar\mapsto c_\hbar(a,b) = e^{-\frac{\hbar}{4}\norm{(a,b)}}$ is continuous for all $a,b\in\mathbb{R}^n$.
\end{enumerate}
This implies (by Thm. 4.4, p. 129 of \citet{HoRi05}) that the maps $c\mathcal{Q}^W_\hbar$ likewise define a strict quantization.  Thus, we can use the maps $c\mathcal{Q}_\hbar^W$ to provide a definition of the Berezin quantization on the minimal algebraic structure of the Weyl algebra.  Since Berezin quantization is positive, and hence continuous, we can extend these maps to unbounded operators defined from the Weyl algebra.

\subsection{Extension to Unbounded Operators}

Our goal is to use the quantization maps $c\mathcal{Q}_\hbar^W$ to analyze the classical limits of unbounded operators like $\Phi_\hbar(a,b)$.  To do so, we note that these operators can be constructed from the unitary generators $W_\hbar(a,b)$ of the Weyl algebra by the formal relation
\begin{align}
\label{eq:pos}
\Phi_\hbar(a,b) := -i\lim_{t\to 0}\frac{W_\hbar(ta,tb) - I}{t}.
\end{align}
This relation holds strictly in the Schr\"odinger representation when the limit is understood in the weak operator topology on $\mathcal{B}(L^2(\mathbb{R}^n))$.  But the limit does not in general converge in the abstract weak topology on $\mathcal{W}(\mathbb{R}^{2n},\hbar\sigma)$.  In the service of our goal of analyzing quantization in a representation manner, we seek a different abstract algebra with a natural topology in which these limits converge.  Then we will be able to use Eq. \ref{eq:pos} as a definition of $\Phi_\hbar(a,b)$ solely in terms of abstract algebraic structure.

To construct such an algebra, we will form the quotient algebra by a certain two-sided ideal.  Our ultimate goal is to find an algebra that allows only for states whose expectation values of Eq. \ref{eq:pos} converge.  It is known\cite{Fe18,Fe18a} that one can limit the collection of states of an algebra if one chooses to quotient by an ideal that is the annihilator of the set of states one wants to focus on.  More precisely, given a C*-algebra $\mathfrak{B}$ and a collection of functionals $V\subseteq \mathfrak{B}^*$, under certain conditions on $V$, one can construct a new C*-algebra $\mathfrak{A}$ whose dual space contains only the functionals in $V$, i.e., $\mathfrak{A}^*\cong V$.  To do so, first let $N(V)$ denote the annihilator of $V$ in $\mathfrak{B}$.  If $N(V)$ is a closed, two-sided ideal in $\mathfrak{B}$, then setting $\mathfrak{A} = \mathfrak{B}/N(V)$ produces a C*-algebra with the desired dual space.

This is relevant to the current circumstance if we focus on the states on the Weyl algebra for which the expectation values of Eq. \ref{eq:pos} converge.  To that end, we focus on the so-called \emph{regular} states and define
\[V_\hbar := \{\omega\in\mathcal{W}(\mathbb{R}^{2n},\hbar\sigma)^*\ |\ t\mapsto \omega(W_\hbar(ta,tb))\text{ is continuous for every }a,b\in\mathbb{R}^n\}.\]
However, in this case, $N(V_\hbar)$ is not a closed, two-sided ideal in $\mathcal{W}(\mathbb{R}^{2n},\hbar\sigma)$ because the latter algebra is simple.  Hence, we move to the bidual  $\mathcal{W}(\mathbb{R}^{2n},\hbar\sigma)^{**}$ and consider the weak* closure $\overline{V}_\hbar\subseteq \mathcal{W}(\mathbb{R}^{2n},\hbar\sigma)^{***}$ of $V_\hbar$, understood now as the regular functionals on the bidual.  It follows that the annihilator $N(\overline{V}_\hbar)$ in $\mathcal{W}(\mathbb{R}^{2n},\hbar\sigma)^{**}$ is now a closed, two-sided ideal.  Hence, we can complete the construction by defining a quotient C*-algebra $\mathfrak{A}_\hbar := \mathcal{W}(\mathbb{R}^{2n},\hbar\sigma)^{**}/N(\overline{V}_\hbar)$.

It follows that $\mathfrak{A}_\hbar^*\cong \overline{V}_\hbar$.  Moreover, the algebra $\mathfrak{A}_0$ is *-isomorphic to the algebra $B_R(\mathbb{R}^{2n})$ of bounded universally Radon measurable functions and the algebras $\mathfrak{A}_\hbar$ are *-isomorphic to $\mathcal{B}(L^2(\mathbb{R}^n))$ for each $\hbar>0$.\citep{Fe18a}  Thus, we have canonical projection (quotient) maps $p_0: AP(\mathbb{R}^{2n})^{**}\to B_R(\mathbb{R}^{2n})$ and $p_\hbar:\mathcal{W}(\mathbb{R}^{2n},\hbar\sigma)^{**}\to\mathcal{B}(L^2(\mathbb{R}^n)$.  And even further, we have $B_R(\mathbb{R}^{2n})\cong C_0(\mathbb{R}^{2n})^{**}$ and $\mathcal{B}(L^2(\mathbb{R}^n))\cong \mathcal{K}(L^2(\mathbb{R}^n))^{**}$ so that both algebras are W*-algebras carrying natural weak* topologies.  The families $t\mapsto p_\hbar(W_\hbar(ta,tb))$ are weak* continuous in $\mathfrak{A}_\hbar$ for all $\hbar\in[0,1]$, so the limit in Eq. \ref{eq:pos} is well-defined in the weak* topology.

We are now in a position to consider the functions $\Phi_0(a,b)$ in the domain of our quantization maps.  To do so, we continuously extend $c\mathcal{Q}^W_\hbar$ in the weak topology to a map $c\tilde{\mathcal{Q}}^W_\hbar:AP(\mathbb{R}^{2n})^{**}\to \mathcal{W}(\mathbb{R}^{2n},\hbar\sigma)^{**}$.  We have the following corollary of Prop. \ref{prop:schr}.
\begin{cor}
\label{cor:schr}
For any $f\in AP(\mathbb{R}^{2n})^{**}$, $p_\hbar\circ c\tilde{\mathcal{Q}}_\hbar^W(f) = \tilde{\mathcal{Q}}^B_\hbar\circ p_0(f)$.  In other words, the diagram in Fig. \ref{fig:2} commutes.

\begin{figure}[h]
  \centerline{
\xymatrix{ {AP(\mathbb{R}^{2n})^{**}} \ar[dd]_{p_0} \ar[rr]^{c\tilde{\mathcal{Q}}^W_\hbar} & & {\mathcal{W}(\mathbb{R}^{2n},\hbar\sigma)^{**}} \ar[dd]^{p_\hbar}\\
& & \\
{B_R(\mathbb{R}^{2n})} \ar[rr]_{\tilde{\mathcal{Q}}^B_\hbar} & & {\mathcal{B}(L^2(\mathbb{R}^n))}
}
}
  \caption{Commutative diagram for Cor. \ref{cor:schr}.}
  \label{fig:2}
\end{figure}
\end{cor}
\noindent This informs us that the map $p_\hbar\circ c\tilde{\mathcal{Q}}_\hbar^W$, which we emphasize can be defined in terms of abstract algebraic structure, is a positive quantization map equivalent to Berezin quantization on $\mathbb{R}^{2n}$.  Thus, $p_\hbar\circ c\tilde{\mathcal{Q}}_\hbar^W$ extends continuously to the entire map $\tilde{\mathcal{Q}}^B_\hbar: \tilde{C}_0(\mathbb{R}^{2n})\to\tilde{\mathcal{K}}(L^2(\mathbb{R}^n))$.  We can understand $\Phi_0(a,b)$ to be defined in the domain $\tilde{C}_0(\mathbb{R}^{2n})$ and $\Phi_\hbar(a,b)$ to be defined in the range $\tilde{\mathcal{K}}(L^2(\mathbb{R}^n))$ both via Eq. \ref{eq:pos}, where the limits are in the abstract weak* topologies.

Finally, we note that the conditions of a strict quantization extend to the unbounded operators $\Phi_\hbar(a,b)$.  The results mentioned now are familiar consequences of Eq. \ref{eq:pos} and the algebraic relations in the Weyl algebra. We establish them explicitly in the more general setting of the next section, but we state them here already.  First, the quantization map assigns $\tilde{\mathcal{Q}}_\hbar^B(\Phi_0(a,b)) = \Phi_\hbar(a,b)$.  Second, the canonical commutation relations are satisfied:
\[[\tilde{\mathcal{Q}}^B_\hbar(\Phi_0(a,b)),\tilde{\mathcal{Q}}^B_\hbar(\Phi_0(a',b'))] = i\hbar\sigma((a,b),(a',b'))I.\]
This implies that Dirac's condition is satisfied in the form
\[\lim_{\hbar\to 0}\norm{\frac{i}{\hbar}[\tilde{\mathcal{Q}}^B_\hbar(\Phi_0(a,b)),\tilde{\mathcal{Q}}^B_\hbar(\Phi_0(a',b'))] -\tilde{\mathcal{Q}}^B_\hbar(\{\Phi_0(a,b),\Phi_0(a',b')\})}_\hbar = 0,\]
with the use of the standard Poisson bracket on $\mathbb{R}^{2n}$.  Furthermore, von Neumann's condition is satisfied in the form
\[\lim_{\hbar\to 0}\norm{\tilde{\mathcal{Q}}_\hbar^B(\Phi_0(a,b))\tilde{\mathcal{Q}}_\hbar^B(\Phi_0(a',b')) - \tilde{\mathcal{Q}}_\hbar^B(\Phi_0(a,b)\Phi_0(a',b'))}_\hbar = 0.\]
Thus, there is a strong sense in which the functions $\Phi_0(a,b)$ can be understood as the classical limits of the operators $\Phi_\hbar(a,b)$.

\section{Generalization to Field Theories}
\label{sec:inf}

Suppose now that $E$ is an infinite dimensional vector space with a symplectic form $\sigma$.  This is the case when $E$ is the test function space for any free Bosonic field theory whose phase space $E'$ is a linear space.  Although the integral formulas defining Berezin quantization in the previous section are no longer meaningful in this context, we proceed to construct an analogous positive quantization, which can likewise be extended to unbounded operators.

\subsection{Positive Quantization}

We start with the Weyl quantization maps $\mathcal{Q}^W_\hbar$, which are well defined even in the infinite-dimensional setting, and we aim to define quantization factors in the spirit of the previous section.  We require a norm on $E$, which may be determined as follows.  Suppose we are given a \emph{complex structure} $J: E\to E$ compatible with $\sigma$---that is, a linear map satisfying
\begin{enumerate}[(i)]
    \item $\sigma(JF,JG) = \sigma(F,G)$;
    \item $\sigma(F,JF)\geq 0$; and
    \item $J^2 = -I$
\end{enumerate}
for all $F,G\in E$.  In general, there is not a unique such complex structure; we will see concrete examples below.  A complex structure can be used to define a complex inner product
\[\alpha_J(F,G) := \sigma(F,JG) + i\sigma(F,G)\]
for all $F,G\in E$.  This inner product $\alpha_J$ allows us to define quantization factors $c^J_\hbar: E\to \mathbb{R}_+$ by $c^J_\hbar(F) := e^{-\frac{\hbar}{4}\alpha_J(F,F)}$.  These quantization factors satisfy the same conditions (a)-(c) of the previous section.  Now, in analogy with the previous section, we define new quantization maps $\mathcal{Q}_\hbar^J: \Delta(E,0)\to \mathcal{W}(E,\hbar\sigma)$ by the linear extension of
\[\mathcal{Q}_\hbar^J(W_0(F)):= c^J_\hbar(F)\mathcal{Q}_\hbar^W(W_0(F)).\]
For any choice of complex structure $J$, this defines a strict quantization equivalent to $\mathcal{Q}^W_\hbar$ in the sense that (See Thm. 4.6, p. 131 of \citet{HoRi05})
\[\lim_{\hbar\to 0}\norm{\mathcal{Q}_\hbar^W(A) - \mathcal{Q}_\hbar^J(A)}_\hbar = 0\]
for all $A\in \Delta(E,0)$.  It follows that the strict quantizations defined for different choices of complex structure $J$ and $J'$ are also all equivalent in this same sense:
\[\lim_{\hbar\to 0}\norm{\mathcal{Q}_\hbar^J(A) - \mathcal{Q}_\hbar^{J'}(A)}_\hbar = 0\]
for all $A\in\Delta(E,0)$.

One can show that the quantization maps $\mathcal{Q}_\hbar^J$ possess some of the same virtues as the Berezin quantization maps of the previous section.

\begin{prop}
If $J$ is a complex structure compatible with the symplectic form $\sigma$, then the map $\mathcal{Q}_\hbar^J: \Delta(E,0)\to\mathcal{W}(E,\hbar\sigma)$ is positive.
\end{prop}

\begin{proof}
Suppose $C\in \Delta(E,0)$ is a positive element.  Then $C = A^*A$ for some $A = \sum_k z_k W_0(F_k)\in \Delta(E,0)$.  We have
\begin{align*}
\mathcal{Q}_\hbar^J(&A^*A) = \sum_{j,k} \overline{z}_j z_k \exp\bigg(-\frac{\hbar}{4}\alpha_J(F_j + F_k,F_j + F_k)\bigg) W_\hbar(F_k - F_j)\\
&=\sum_{j,k} \overline{z}_j z_k\exp\bigg(-\frac{\hbar}{4}\Big(\alpha_J(F_j,F_j) + \alpha_J(F_k,F_k) + 2\sigma(F_j,JF_k)\Big)-\frac{i\hbar}{2}\sigma(F_j,F_k)\bigg)W_\hbar(-F_j)W_\hbar(F_k).\\
\end{align*}
Letting $y_k = e^{-\frac{\hbar}{4}\alpha_J(F_k,F_k)}z_k$, it follows that
\begin{align*}
    \mathcal{Q}_\hbar^J(A^*A) &=
    \sum_j\sum_k \overline{y}_jy_k\exp\bigg(-\frac{\hbar}{2}\Big(\sigma(F_j,JF_k) + i\sigma(F_j,F_k)\Big)\bigg)W_\hbar(F_j)^*W_\hbar(F_k)\\
    &= \sum_j\sum_k \overline{y}_jy_k\exp\Big(-\frac{\hbar}{2}\alpha_J(F_j,F_k) \Big)W_\hbar(F_j)^*W_\hbar(F_k)
\end{align*}
Since $\alpha_J$ is a complex inner product, the matrix $a_{j,k}: =\alpha_J(F_j,F_k)$ is positive, and moreover, since entrywise exponentiation preserves positivity, the matrix $b_{j,k}:= \exp(a_{j,k})$ is also positive.  It then follows from a generalization of the Schur product theorem due to \citet{SuSu16} (Prop. 1.3) that $\mathcal{Q}_\hbar^J(A^*A)$ is a positive element in $\mathcal{W}(E,\hbar\sigma)$.
\end{proof}

The positivity of $\mathcal{Q}_\hbar^J$ implies its continuity in the norm and weak topologies (see Prop. 1.3.7 of \citeauthor{La98b}\cite{La98b}, p. 47), which means that it can be continuously extended to the completions of its domain and range.  As in the previous section, we want to use these extended quantization maps to analyze field operators of the form
\begin{align}
\label{eq:fields}
\Phi_\hbar(F):=-i\lim_{t\to0} \frac{W_\hbar(tF)-I}{t}.
\end{align}
However, these limits again do not converge in the weak topology.  So we must perform the construction of the previous section to arrive at a new algebra allowing for this definition.

To construct such an algebra, we again quotient out by a certain two-sided ideal given by the annihilator of a desired set of states.  We again focus on the regular states for which the expectation values of Eq. \ref{eq:fields} converge by defining the set of regular states as
\[V_\hbar := \{\omega\in\mathcal{W}(E,\hbar\sigma)^*\ |\ t\mapsto\omega(tF)\text{ is continuous for every }F\in E\}.\]
Just as before, $N(V_\hbar)$ is not a closed, two-sided ideal because $\mathcal{W}(E,\hbar\sigma)$ is simple.  Instead, we use the strategy of the previous section by passing to the bidual $\mathcal{W}(E,\hbar\sigma)^{**}$ and letting $\overline{V}_\hbar$ be the weak* closure of $V_\hbar$ in $\mathcal{W}(E,\hbar\sigma)^{***}$.  Then $N(\overline{V}_\hbar)$ is a closed, two sided ideal, so we can define a C*-algebra $\mathfrak{A}_\hbar:= \mathcal{W}(E,\hbar\sigma)^{**}/N(\overline{V}_\hbar)$ exactly as before.

However, since $E$ is now infinite-dimensional and so fails to be locally compact, the structure of these algebras $\mathfrak{A}_\hbar$ is not as tractable and we have much less information than in the previous section.  Still, we can show that the algebras $\mathfrak{A}_\hbar$ are W*-algebras with an appropriate weak* topology.

\begin{prop}
The algebras $\mathfrak{A}_\hbar$ are W*-algebras with preduals given by $(\mathfrak{A}_\hbar)_* = V_\hbar$.
\end{prop}

\begin{proof}
First, let $\pi_U$ denote the universal representation of $\mathcal{W}(E,\hbar\sigma)$.  We will consider the direct sum representation $\pi := \bigoplus_{\omega\in I}\pi_\omega$ for $I = V_\hbar\cap \mathcal{S}(\mathcal{W}(E,\hbar\sigma))$, where $\pi_\omega$ is the GNS representation for the state $\omega$ and $\mathcal{S}(\mathcal{W}(E,\hbar\sigma))$ denotes the state space of the Weyl algebra.  It follows from Thm. 10.1.12 of \citet[][]{KaRi97} (p. 719) that there is a projection $P$ in the center of $\overline{\pi_U(\mathcal{W}(E,\hbar\sigma)}$, where the closure is in the weak operator topology, such that $\overline{\pi_U(\mathcal{W}(E,\hbar\sigma)}P$ is *-isomorphic to $\overline{\pi(\mathcal{W}(E,\hbar\sigma)}$, the latter of which is *-isomorphic to $\mathfrak{A}_\hbar$.  By Prop. 5.5.6 of \citet[][]{KaRi97} (p. 335), the algebra $\overline{\pi_U(\mathcal{W}(E,\hbar\sigma)}P$ is a W*-algebra, which implies that $\mathfrak{A}_\hbar$ is a W*-algebra.  Moreover, by Prop. 5 of \citet[][]{Ho97} (p. 15) it follows that $(\mathfrak{A}_\hbar)_* = V_\hbar$.
\end{proof}

\noindent This implies that $\mathcal{Q}_\hbar^J$ extends continuously to a map whose codomain is the weak* completion $\tilde{\mathfrak{A}}_\hbar$, which we now denote $\tilde{Q}_\hbar^J: \tilde{\mathfrak{A}}_0\to\tilde{\mathfrak{A}}_\hbar$.  The field operators are well-defined in these completed algebras via Eq. \ref{eq:fields} with the limit now understood in the weak* topology.  Now we can use the maps $\tilde{Q}_\hbar^J$ to analyze the classical limits of unbounded field operators.

\subsection{Extension to Unbounded Operators}

First, we note that the familiar facts about the field operators $\Phi_\hbar(F)$ follow from what has been said so far.  We present proofs here to emphasize the fact that these statements can be both expressed and derived in the bare algebraic setting we have outlined.

\begin{lemma}
For all $F\in E$, $\Phi_\hbar(F)$ is self-adjoint.
\end{lemma}

\begin{proof}
For any $F\in E$, we have
\begin{align*}
    (\Phi_\hbar(F))^* &= \bigg(-i\lim_{t\to 0} \frac{W_\hbar(tF) - I}{t}\bigg)^*
    = i\lim_{t\to 0}\frac{W_\hbar(-tF) - I}{t}
    = - i\lim_{s\to 0}\frac{W_\hbar(sF)-I}{s} = \Phi_\hbar(F).
\end{align*}
In the third line, we make the replacement $s = -t$.
\end{proof}

\begin{lemma}
\label{lem:fields}
For all $F\in E$, $\tilde{\mathcal{Q}}_\hbar^J(\Phi_0(F)) = \Phi_\hbar(F)$.  In other words, the diagram in Fig. \ref{fig:3} commutes.

\begin{figure}[h]
\centerline{
\xymatrix{ & {E} \ar[dl]_{\Phi_0} \ar[dr]^{\Phi_\hbar} & \\
\mathcal{W}_\hbar(E,0) \ar[rr]_{\tilde{\mathcal{Q}}_\hbar^J} & & \mathcal{W}_\hbar(E,\hbar\sigma) 
}
}
  \caption{Commutative diagram for Lemma \ref{lem:fields}.}
  \label{fig:3}
\end{figure}
\end{lemma}

\begin{proof}
For any $F\in E$, we have
\begin{align*}
    \tilde{\mathcal{Q}}_\hbar^J(\Phi_0(F)) &= \tilde{\mathcal{Q}}_\hbar^J\bigg(-i\lim_{t\to 0}\frac{W_0(tF)-I}{t}\bigg)\\
    &= -i\lim_{t\to 0}\frac{e^{-\frac{\hbar}{4}\alpha_J(tF,tF)}W_\hbar(tF)-I}{t}\\
    &= -i\lim_{t\to 0}\frac{e^{-\frac{\hbar}{4}\alpha_J(tF,tF)}W_\hbar(tF) - W_\hbar(tF)}{t} - i\lim_{t\to 0}\frac{W_\hbar(tF)-I}{t}\\
    &= -i\lim_{t\to 0}\frac{(e^{-\frac{\hbar}{4}t^2\alpha_J(F,F)}-1)W_\hbar(tF)}{t} + \Phi_\hbar(F)\\
    &= -i\bigg(\lim_{t\to 0}\frac{e^{-\frac{\hbar}{4}t^2\alpha_J(F,F)}-1}{t}\bigg)\Big(\lim_{t\to 0}W_\hbar(tF)\Big) + \Phi_\hbar(F)\\
    &= -i(0)(I) + \Phi_\hbar(F) = \Phi_\hbar(F). \qedhere
\end{align*}
\end{proof}

\begin{lemma}
\label{lem:comm}
For all $F,G\in E$ and all $n\in\mathbb{N}$, $[\Phi_\hbar(F),\Phi_\hbar(G)^n] = in\hbar\sigma(F,G)\Phi_\hbar(G)^{n-1}$
\end{lemma}

\begin{proof}
First, we compute
\begin{equation*}
\begin{split}
    [\Phi_\hbar(F),&\Phi_\hbar(G)] \\ &\begin{split} =  \bigg(-i\lim_{s\to 0}\frac{W_\hbar(sF)-I}{s}\bigg)\bigg(-i\lim_{t\to 0}&\frac{W_\hbar(tG)-I}{t}\bigg)\\ - &\bigg(-i\lim_{t\to 0}\frac{W_\hbar(tG)-I}{t}\bigg)\bigg(-i\lim_{s\to 0}\frac{W_\hbar(sF)-I}{s}\bigg)
    \end{split}\\
    &= \lim_{s\to 0}\lim_{t\to 0}\frac{1}{st}\Big((W_\hbar(tG)-I)(W_\hbar(sF)-I) - (W_\hbar(sF)-I)(W_\hbar(tG)-I)\Big)\\
    &= \lim_{s\to 0}\lim_{t\to 0}\frac{1}{st}\Big(e^{-\frac{i\hbar}{2}\sigma(tG,sF)} - e^{-\frac{i\hbar}{2}\sigma(sF,tG)}\Big)W_\hbar(sF + tG)\\
    &= \bigg(\lim_{s\to 0}\lim_{t\to 0}\frac{e^{-\frac{i\hbar}{2}st\sigma(G,F)} - e^{-\frac{i\hbar}{2}st\sigma(F,G)}}{st}\bigg)\bigg(\lim_{s\to 0}\lim_{t\to 0}W_\hbar(sF + tG)\bigg)\\
    &= i\hbar\sigma(F,G) I.
    \end{split}
\end{equation*}

Next we proceed by induction.  Suppose $[\Phi_\hbar(F),\Phi_\hbar(G)^k] = ik\hbar\sigma(f,g)(\Phi_\hbar(G))^{k-1}$ for some $k\in\mathbb{N}$.  Then
\begin{align*}
    \Phi_\hbar(F)\Phi_\hbar(G)^{k+1} &= (\Phi_\hbar(F) \Phi_\hbar(G)^k)\Phi_\hbar(G)\\
    &= \Big(ik\hbar\sigma(F,G)\Phi_\hbar(G)^{k-1} + \Phi_\hbar(G)^k\Phi_\hbar(F)\Big)\Phi_\hbar(G)\\
    &=ik\hbar\sigma(F,G)\Phi_\hbar(G)^k + \Phi_\hbar(G)^k\Big(i\hbar\sigma(F,G)I + \Phi_\hbar(G)\Phi_\hbar(F)\Big)\\
    &= i(k+1)\hbar\sigma(F,G)\Phi_\hbar(G)^k + \Phi_\hbar(G)^{k+1}\Phi_\hbar(F),
\end{align*}
which implies $[\Phi_\hbar(F),\Phi_\hbar(G)^{k+1}] = i(k+1)\hbar\sigma(F,G)\Phi_\hbar(G)^k$.
\end{proof}

\begin{lemma}
For all $F\in E$, $\tilde{\mathcal{Q}}_\hbar^J(\Phi_0(F)^2) = \Phi_\hbar(F)^2 + \frac{\hbar}{2}\alpha_J(F,F) I$.
\end{lemma}

\begin{proof}
First, we compute
\begin{align*}
    \tilde{\mathcal{Q}}_\hbar^J(\Phi_0(F)^2) &= \tilde{\mathcal{Q}}_\hbar^J\bigg((-i)^2\lim_{s\to 0}\lim_{t\to 0}\frac{(W_0(sF) - I)(W_0(tF) - I)}{st}\bigg)\\
    &=-\tilde{\mathcal{Q}}_\hbar^J\bigg(\lim_{s\to 0}\lim_{t\to 0}\frac{W_0((s+t)F) - W_0(sF) - W_0(tF) + I}{st}\bigg)\\
    &= -\lim_{s\to 0}\lim_{t\to 0}\frac{e^{-\frac{\hbar}{4}(s+t)^2\alpha_J(F,F)}W_\hbar((s+t)F) - e^{-\frac{\hbar}{4}s^2\alpha_J(F,F)}W_\hbar(sF) - e^{-\frac{\hbar}{4}t^2\alpha_J(F,F)}W_\hbar(tF) + I}{st}.
\end{align*}

Compare this to the value
\begin{align*}
    \tilde{\mathcal{Q}}_\hbar^J(\Phi_0(&F))^2 = (-i)^2\bigg(\lim_{t\to 0}\frac{\mathcal{Q}_\hbar^J(W_0(tF)) - I}{t}\bigg)^2\\
    &= - \lim_{s\to 0}\lim_{t\to 0}\frac{\Big(e^{-\frac{\hbar}{4}s^2\alpha_J(F,F)}W_\hbar(sF) - I\Big)\Big(e^{-\frac{\hbar}{4}t^2\alpha_J(F,F)}W_\hbar(tF) - I\Big)}{st}\\
    &= -\lim_{s\to 0}\lim_{t\to 0}\frac{\Big(e^{-\frac{\hbar}{4}(s^2+t^2)\alpha_J(F,F)}W_\hbar((s+t)F) - e^{-\frac{\hbar}{4}s^2\alpha_J(F,F)}W_\hbar(sF) -  e^{-\frac{\hbar}{4}t^2\alpha_J(F,F)}W_\hbar(tF) + I\Big)}{st}.
\end{align*}

This gives the identity
\begin{align*}
    \tilde{\mathcal{Q}}_\hbar^J(\Phi_0(F)^2) - \tilde{\mathcal{Q}}_\hbar(\Phi_0(F))^2 &= - \bigg(\lim_{s\to 0}\lim_{t\to 0}\frac{ e^{-\frac{\hbar}{4}(s+t)^2\alpha_J(F,F)} - e^{-\frac{\hbar}{4}(s^2+t^2)\alpha_J(F,F)}}{st}W_\hbar((s+t)F)\bigg)\\
    &= \bigg(\lim_{s\to 0}\lim_{t\to 0}\frac{e^{-\frac{\hbar}{4}(s^2+t^2)\alpha_J(F,F)} - e^{-\frac{\hbar}{4}(s+t)^2\alpha_J(F,F)}}{st}\bigg)\bigg(\lim_{s\to 0}\lim_{t\to 0} W_\hbar((s+t)F)\bigg)\\
    &= \frac{\hbar}{2}\alpha_J(F,F) I.
\end{align*}
Thus, $\tilde{\mathcal{Q}}_\hbar^J(\Phi_0(F)^2) = \tilde{\mathcal{Q}}_\hbar^J(\Phi_0(F))^2 + \frac{\hbar}{2}\alpha_J(F,F)I = \Phi_\hbar(F)^2 + \frac{\hbar}{2}\alpha_J(F,F)I$.
\end{proof}

We would like to extend the conditions of a strict quantization to even unbounded operators such as $\Phi_\hbar(F)$.  However, since we have extended the quantization map in the weak topology, the resulting notion of approximation in the classical limit is significantly weaker than the norm approximations in a strict quantization.  We do at least have a notion of approximation pointwise on each state, as follows.  Fix some choice of $H\in[0,1]$ and an arbitrary functional $\omega_H\in V_{H}$.  We construct the ``constant" section of linear functionals $\{\omega_\hbar\in V_\hbar\}_{\hbar\in[0,1]}$ through the point $\omega_H$ as the continuous extension of
\[\omega_\hbar: \mathcal{Q}_\hbar^J(A)\mapsto \omega_H(\mathcal{Q}_{H}^J(A))\in\mathbb{C}\]
for each $\hbar\in[0,1]$ and each $A\in\Delta(E,0)$.  Then for any $A,B\in\tilde{\mathfrak{A}}_0$ and any $\epsilon>0$, there is an $\hbar'\in (0,1]$ such that for all $\hbar<\hbar'$,
\[\Big| \omega_\hbar\Big(\tilde{\mathcal{Q}}_\hbar^J(A)\tilde{\mathcal{Q}}_\hbar^J(B) - \tilde{\mathcal{Q}}_\hbar^J(AB)\Big)\Big|<\epsilon\]
when $AB$ and $\tilde{\mathcal{Q}}_\hbar^J(A)\tilde{\mathcal{Q}}^J_\hbar(B)$ are well-defined.  This approximation is of course much weaker than one would like.  However, we show next that the preliminary lemmas just stated imply that the classical limits of field operators in particular satisfy a stronger approximation given by Dirac's condition and von Neumann's condition for a strict quantization.  This follows because although the field operators are unbounded and the norm is not defined on them, the relevant differences of operators are bounded and so the conditions are meaningful exactly as stated.  In what follows, we understand the Poisson bracket to be defined as in \citet[][Eq. 2.15, p. 11]{HoRiSc08}; cf. the Peierls bracket as defined in \citet{FrRe15} and \citeauthor{Re16}.\citep{Re16}  

\begin{prop}
\label{prop:lims}
For all $F,G\in E$,
\begin{enumerate}
    \item $\lim_{\hbar\to 0}\norm{\frac{i}{\hbar}[\Phi_\hbar(F),W_\hbar(G)] - \tilde{\mathcal{Q}}^J_\hbar(\{\Phi_0(F),W_0(G)\})}_\hbar = 0$.
    \item $\lim_{\hbar\to 0}\norm{\Phi_\hbar(F)W_\hbar(G) - \tilde{\mathcal{Q}}_\hbar^J(\Phi_0(F)W_0(G))}_\hbar = 0$
    \item $\lim_{\hbar\to 0}\norm{\frac{i}{\hbar}[\Phi_\hbar(F),\Phi_\hbar(G)] - \tilde{\mathcal{Q}}^J_\hbar(\{\Phi_0(F),\Phi_0(G)\})}_\hbar = 0$.
    \item $\lim_{\hbar\to 0}\norm{\Phi_\hbar(F)\Phi_\hbar(G) - \tilde{\mathcal{Q}}_\hbar^J(\Phi_0(F)\Phi_0(G))}_\hbar = 0$
\end{enumerate}
\end{prop}

\begin{proof}

\begin{enumerate}
    \item First, we have $\{\Phi_0(F),W_0(G)\} = i\sigma(G,F)W_0(G)$, which implies \[\tilde{\mathcal{Q}}_\hbar^J(\{\Phi_0(F),W_0(G)\} = i\sigma(G,F)e^{-\frac{\hbar}{4}\alpha_J(G,G)}W_\hbar(G).\]
    Further,
    \begin{align*}
        \Phi_\hbar(F)W_\hbar(G) &= \bigg(-i\lim_{t\to 0}\frac{W_\hbar(tF)-I}{t}\bigg)W_\hbar(G)\\
        &= W_\hbar(G)\bigg(-i\lim_{t\to 0}\frac{e^{i\hbar\sigma(G,tF)}W_\hbar(tF)-I}{t}\bigg)\\
        &=W_\hbar(G)\bigg(\Big(-i\lim_{t\to 0}\frac{e^{i\hbar\sigma(G,tF)}-1}{t}\Big)I -i\lim_{t\to 0}\frac{W_\hbar(tF)-I}{t}\bigg)\\
        &= W_\hbar(G)((-i)(i)\hbar\sigma(G,F)I + \Phi_\hbar(F))\\
        &= W_\hbar(G)\Phi_\hbar(F) + \hbar\sigma(G,F)W_\hbar(G),
    \end{align*}
which implies
\[\frac{i}{\hbar}[\Phi_\hbar(F),W_\hbar(G)] = i\sigma(G,F)W_\hbar(G)\]
and hence
\begin{align*}
    \lim_{\hbar\to 0}\norm{\frac{i}{\hbar}[\Phi_\hbar(F),W_\hbar(G)] &- \tilde{\mathcal{Q}}^J_\hbar(\{\Phi_0(F),W_0(G)\})}_\hbar\\
    &= \lim_{\hbar\to 0}\norm{i\sigma(G,F) - i\sigma(G,F)e^{-\frac{\hbar}{4}\alpha_J(G,G)}W_\hbar(G)}_\hbar\\
    &= \lim_{\hbar\to 0}\abs{i\sigma(G,F)(1-e^{-\frac{\hbar}{4}\alpha_J(G,G)})} = 0.
\end{align*}
\item We have
\[\Phi_0(F)W_0(G) = -i\lim_{t\to 0}\frac{W_0(tF+G)-W_0(G)}{t}\]
so that
\begin{equation*}
\begin{split}    \tilde{\mathcal{Q}}_\hbar^J(\Phi_0(F)W_0(G)) &= -i\lim_{t\to 0}\frac{e^{-\frac{\hbar}{4}\alpha_J(tF+G,tF+G)}W_\hbar(tF+G) - e^{-\frac{\hbar}{4}\alpha_J(G,G)}W_\hbar(G)}{t}\\ &=\bigg(-i\lim_{t\to 0}\frac{e^{-\frac{\hbar}{4}\alpha_J(tF+G,tF+G)}e^{\frac{i\hbar}{2}\sigma(tF,G)}W_\hbar(tF) - e^{-\frac{\hbar}{4}\alpha_J(G,G)}I}{t}\bigg)W_\hbar(G)\\
     &\begin{split}= \bigg(\Big(-i\lim_{t\to 0}&\frac{e^{-\frac{\hbar}{4}\alpha_J(tF+G,tF+G)}e^{\frac{i\hbar}{2}\sigma(tF,G)}-1}{t}\Big)I \\&- i e^{-\frac{\hbar}{4}\alpha_J(G,G)}\Big(\lim_{t\to0}\frac{W_\hbar(tF)-I}{t}\Big)\bigg)W_\hbar(G)\end{split} \\
     &=\Big(\frac{i\hbar}{2}\sigma(F,JG)I + e^{-\frac{\hbar}{4}\alpha_J(G,G)}\Phi_\hbar(F)\Big)W_\hbar(G).
\end{split}
\end{equation*}
It follows that
\begin{equation*}
    \begin{split}
        \lim_{\hbar\to 0}\norm{\Phi_\hbar(F)W_\hbar(G) - \tilde{\mathcal{Q}}_\hbar^J(\Phi_0(F)W_0(G))}_\hbar &= \lim_{\hbar\to 0} \abs{-\frac{i\hbar}{2}\sigma(F,JG)} = 0.
    \end{split}
\end{equation*}

\item This follows immediately from Lemma \ref{lem:comm} together with the fact that \[\{\Phi_0(F),\Phi_0(G)\} = -\sigma(F,G).\]

\item We have
\begin{align*}
   &\tilde{\mathcal{Q}}_\hbar^J(\Phi_0(F)\Phi_0(G))\\ &= (-i)^2\tilde{\mathcal{Q}}_\hbar^J\bigg(\lim_{t\to0}\lim_{s\to 0}\frac{W_0(tF+sG) - W_0(tF) - W_0(sG) + I}{st}\bigg)\\
    &=(-i)^2\lim_{t\to 0}\lim_{s\to 0}\frac{e^{-\frac{\hbar}{4}\alpha_J(tF+sG,tF+sG)}W_\hbar(tF+sG) - e^{-\frac{\hbar}{4}\alpha_J(tF,tF)}W_\hbar(tF) - e^{-\frac{\hbar}{4}\alpha_J(sG,sG)}W_\hbar(sG)+I}{st}
\end{align*}
and
\begin{equation*}
\begin{split}
    &\tilde{\mathcal{Q}}^J_\hbar(\Phi_0(F))\tilde{\mathcal{Q}}_\hbar^J(\Phi_0(G))\\
    &= (-i)^2\bigg(\lim_{t\to 0}\frac{e^{-\frac{\hbar}{4}\alpha_J(tF,tF)}W_\hbar(tF)-I}{t}\bigg)\bigg(\lim_{s\to 0}\frac{e^{-\frac{\hbar}{4}\alpha_J(sG,sG)}W_\hbar(sG)-I}{s}\bigg)\\
    &\begin{split}= (-i)^2\lim_{t\to 0}\lim_{s\to 0}\frac{1}{st}\bigg(e^{-\frac{\hbar}{4}(\alpha_J(tF,tF) + \alpha_J(sG,sG))}&e^{-\frac{i\hbar}{2}\sigma(tF,sG)}W_\hbar(tF+sG) \\&- e^{-\frac{\hbar}{4}\alpha_J(tF,tF)}W_\hbar(tF) - e^{-\frac{\hbar}{4}\alpha_J(sG,sG)}W_\hbar(sG)+I\bigg),\end{split}
    \end{split}
\end{equation*}
which implies
\begin{equation*}
    \begin{split}
        \tilde{\mathcal{Q}}_\hbar^J(\Phi_0(F))&\tilde{\mathcal{Q}}_\hbar^J(\Phi_0(G)) - \tilde{\mathcal{Q}}_\hbar^J(\Phi_0(F)\Phi_0(G))\\ &= (-i)^2\lim_{t\to 0}\lim_{s\to 0} \frac{1}{st}e^{-\frac{\hbar}{4}(\alpha_J(tF,tF) + \alpha_J(sG,sG))}\bigg(e^{-\frac{i\hbar}{2}st\sigma(F,G)} - e^{-\frac{\hbar}{2}st\alpha_J(F,G)}\bigg)W_\hbar(tF+sG)\\
        &=\lim_{t\to 0}\lim_{s\to 0}e^{-\frac{i\hbar}{2}st\sigma(F,G)}\bigg(\frac{e^{-\frac{\hbar}{2}st\sigma(F,JG)}-1}{st}\bigg)I\\
        &=-\frac{\hbar}{2}\sigma(F,JG)
    \end{split}
\end{equation*}
and hence
\begin{align*}
    \lim_{\hbar\to 0}\norm{\Phi_\hbar(F)\Phi_\hbar(G) - \tilde{\mathcal{Q}}_\hbar^J(\Phi_0(F)\Phi_0(G))}_\hbar &= \lim_{\hbar\to 0}\abs{-\frac{\hbar}{2}\sigma(F,JG)} = 0. \qedhere
\end{align*}
\end{enumerate}

\end{proof}

\noindent This establishes a strong sense in which $\Phi_0(F)$ is the classical limit of $\Phi_\hbar(F)$.

Suppose further that we are given a complex structure $J_0$ compatible with $\sigma$, which may be distinct from the complex structure $J$ used to define the quantization map.  In the next section, we will consider two such possible complex structures. It is important to note that facts about the classical limits of quantities defined in the quantum theory via a complex structure do not depend on which complex structure is used in the definition of the quantization map.  We can use a complex structure to define $J_0$-creation and $J_0$-annihilation operators abstractly by
\begin{align*}
a_\hbar^{J_0}(F):= \frac{1}{\sqrt{2}}\Big(\Phi_\hbar(F) + i\Phi_\hbar(J_0F)\Big) &&
\big(a_\hbar^{J_0}(F)\big)^*:= \frac{1}{\sqrt{2}}\Big(\Phi_\hbar(F) - i\Phi_\hbar(J_0F)\Big).
\end{align*}
Notice that it follows immediately from Lemma \ref{lem:fields} that $\mathcal{Q}_\hbar^J(a_0^{J_0}(F)) = a_\hbar^{J_0}(F)$ even when $J$ and $J_0$ are distinct.  Similarly, the $J_0$-creation and $J_0$-annihilation operators can be used to abstractly define the $J_0$-number operators
\begin{align*}
    N_\hbar^{J_0}(F) := \big(a_\hbar^{J_0}(F)\big)^*a_\hbar^{J_0}(F).
\end{align*}
Although $\mathcal{Q}_\hbar^J(N_0^{J_0}(F))\neq N_\hbar^{J_0}(F)$, we can still show a sense in which $N_0^{J_0}(F)$ is the classical limit of $N_\hbar^{J_0}(F)$.

\begin{prop} For all $F\in E$, $\lim_{\hbar\to 0}\norm{\tilde{\mathcal{Q}}_\hbar^J(N_0^{J_0}(F)) - N^{J_0}_\hbar(F)}_\hbar = 0$.
\end{prop}

\begin{proof}
Note that
\[N_\hbar^{J_0}(F) = \frac{1}{2}\Big(\Phi_\hbar(F)^2 + \Phi_\hbar(J_0F)^2 + i[\Phi_\hbar(F),\Phi_\hbar(J_0F)]\Big)\]
and
\begin{align*}
    \tilde{\mathcal{Q}}_\hbar^J(N_0^{J_0}(F)) &= \frac{1}{2}\tilde{\mathcal{Q}}_\hbar^J\Big(\Phi_0(F)^2 + \Phi_0(J_0F)^2\Big)\\
    &= \frac{1}{2}\Big(\Phi_\hbar(F)^2 + \Phi_\hbar(J_0F)^2 + \frac{\hbar}{2}(\alpha_J(F,F) + \alpha_J(J_0F,J_0F))\Big)\\
    &= \frac{1}{2}\Big(\Phi_\hbar(F)^2 + \Phi_\hbar(J_0F)^2 + \hbar\alpha_J(F,F)\Big)
\end{align*}
so since $[\Phi_\hbar(F),\Phi_\hbar(J_0F)] = i\hbar\alpha_{J_0}(F,F)$, it follows that
\begin{align*}
    \lim_{\hbar\to 0}\norm{\tilde{\mathcal{Q}}_\hbar^J(N_0^{J_0}(F)) - N_\hbar^{J_0}(F)}_\hbar &= \lim_{\hbar\to 0}\abs{\frac{\hbar}{2}(\alpha_J(F,F)+\alpha_{J_0}(F,F))} = 0. \qedhere
\end{align*}
\end{proof}

Furthermore, one can show that Dirac's condition and von Neumann's condition hold for some combinations of creation (or annihilation) operators and number operators.

\begin{prop}
For any $F,G\in E$,
\begin{enumerate}
\item $\lim_{\hbar\to 0}\norm{\frac{i}{\hbar}[a_\hbar^{J_0}(F),W_\hbar(G)] - \tilde{\mathcal{Q}}_\hbar^J(\{a_0^{J_0}(F),W_0(G)\})}_\hbar = 0$.
\item $\lim_{\hbar\to 0} \norm{a_\hbar^{J_0}(F)W_\hbar(G) - \tilde{\mathcal{Q}}^J_\hbar(a_0^{J_0}(F)W_0(G))}_\hbar = 0$.
\item $\lim_{\hbar\to 0}\norm{\frac{i}{\hbar}[a_\hbar^{J_0}(F),\Phi_\hbar(G)] - \tilde{\mathcal{Q}}_\hbar^J(\{a_0^{J_0}(F),\Phi_0(G)\})}_\hbar = 0$.
\item $\lim_{\hbar\to 0} \norm{a_\hbar^{J_0}(F)\Phi_\hbar(G) - \tilde{\mathcal{Q}}^J_\hbar(a_0^{J_0}(F)\Phi_0(G))}_\hbar = 0$.
\item $\lim_{\hbar\to 0}\norm{\frac{i}{\hbar}\Big[\big(a_\hbar^{J_0}(F)\big)^*,a^{J_0}_\hbar(G)\Big] - \tilde{\mathcal{Q}}_\hbar^J\Big(\Big\{\big(a_0^{J_0}(F)\big)^*,a^{J_0}_0(G)\Big\}\Big)}_\hbar = 0$.
\item $\lim_{\hbar\to 0} \norm{\big(a_\hbar^{J_0}(F)\big)^*a^{J_0}_\hbar(G) - \tilde{\mathcal{Q}}^J_\hbar\Big(\big(a_0^{J_0}(F)\big)^*a^{J_0}_0(G)\Big)}_\hbar = 0$.
\end{enumerate}
\end{prop}

\begin{proof}
This follows immediately from Prop.
\ref{prop:lims} along with the linearity of the creation and annihilation operators with respect to the field operators.
\end{proof}

\begin{prop}
For all $F,G\in E$
\begin{enumerate}
    \item $\lim_{\hbar\to 0}\norm{\frac{i}{\hbar}[N_\hbar^{J_0}(F),W_\hbar(G)] - \tilde{\mathcal{Q}}_\hbar^J(\{N_0^{J_0}(F),W_0(G)\})}_\hbar = 0$.

    \item $\lim_{\hbar\to 0}\norm{\frac{i}{\hbar}[N_\hbar^{J_0}(F),\Phi_\hbar(G)] - \tilde{\mathcal{Q}}_\hbar^{J_0}(\{N_0^{J_0}(F),\Phi_0(G)\})}_\hbar = 0$.

    \item $\lim_{\hbar\to 0}\norm{\frac{i}{\hbar}[N_\hbar^{J_0}(F),a_\hbar^{J_0}(G)] - \tilde{\mathcal{Q}}_\hbar^J(\{N_0^{J_0}(F),a_0^{J_0}(G)\})}_\hbar = 0.$

    \item $\lim_{\hbar\to 0}\norm{\frac{i}{\hbar}[N_\hbar^{J_0}(F),N_\hbar^{J_0}(G)] - \tilde{\mathcal{Q}}_\hbar^J(\{N_0^{J_0}(F),N_0^{J_0}(G)\})}_\hbar = 0.$

\end{enumerate}
\end{prop}

\begin{proof}

\begin{enumerate}
    
    \item We have
    \[\{N_0^{J_0}(F),W_0(G)\} = i\big(\sigma(G,F)\Phi_0(F)W_0(G) + \sigma(G,J_0F)\Phi_0(J_0F)W_0(G)\big)\]
    and
    \begin{equation*}
        \begin{split}
        [N_\hbar^{J_0}(F),W_\hbar(&G)]= \frac{1}{2}[\Phi_\hbar(F)^2 + \Phi_\hbar(J_0F)^2,W_\hbar(G)]\\
        &\begin{split}=\frac{1}{2}\Big([\Phi_\hbar(F),W_\hbar(G)]&\Phi_\hbar(F) + \Phi_\hbar(F)[\Phi_\hbar(F),W_\hbar(G)]\\ &+ [\Phi_\hbar(J_0F),W_\hbar(G)]\Phi_\hbar(J_0F) + \Phi_\hbar(J_0F)[\Phi_\hbar(J_0F)),W_\hbar(G)]\Big)\end{split}\\
        &\begin{split}= \frac{\hbar}{2}\sigma(G,F)\Big(W_\hbar(G)&\Phi_\hbar(F) + \Phi_\hbar(F)W_\hbar(G)\Big)\\& + \frac{\hbar}{2}\sigma(G,J_0F)\Big(W_\hbar(G)\Phi_\hbar(J_0F) + \Phi_\hbar(J_0F)W_\hbar(G)\Big).
        \end{split}
        \end{split}
    \end{equation*}
        Then (2) of Prop. \ref{prop:lims} implies the result.

    \item We have
    \[\{N_0^{J_0}(F),\Phi_0(G)\} = \sigma(G,F)\Phi_0(F) + \sigma(G,J_0F)\Phi_0(J_0F)\]
    and
    \begin{equation*}
        \begin{split}
         [N_\hbar^{J_0}(F),\Phi_\hbar(G)]&=\frac{1}{2}[\Phi_\hbar(F)^2 + \Phi_\hbar(J_0F)^2,\Phi_\hbar(G)]\\
         &= i\hbar\Big(\sigma(F,G)\Phi_\hbar(F) + \sigma(J_0F,G)\Phi_\hbar(J_0F)\Big).
        \end{split}
    \end{equation*}
    Then Lemma \ref{lem:fields} implies the result.
    
    \item This follows from (2) and the linearity of $a_\hbar^{J_0}(G)$ with respect to the fields.
    
    \item We have
    \begin{align*}
        \{N_0^{J_0}(F),N_0^{J_0}(G)\} = \sigma(G,F)&\Phi_0(F)\Phi_0(G) + \sigma(J_0G,F)\Phi_0(F)\Phi_0(J_0G)\\ &+ \sigma(G,J_0F)\Phi_0(J_0F)\Phi_0(G) + \sigma(G,F)\Phi_0(J_0F)\Phi_0(J_0G)
    \end{align*}
    and
    \begin{equation*}
        \begin{split} [N_\hbar^{J_0}(F),N_\hbar^{J_0}(G)] &= \frac{1}{4}[\Phi_\hbar(F)^2 + \Phi_\hbar(J_0F)^2,\Phi_\hbar(G)^2 + \Phi_\hbar(J_0G)^2]\\
            &\begin{split}     =\frac{i\hbar}{2}&\sigma(F,G)\Big(\Phi_\hbar(F)\Phi_\hbar(G) + \Phi_\hbar(G)\Phi_\hbar(F)\Big)\\
            &+ \frac{i\hbar}{2}\sigma(F,J_0G)\Big(\Phi_\hbar(F)\Phi_\hbar(J_0G) + \Phi_\hbar(J_0G)\Phi_\hbar(F)\Big)\\ &\ +\frac{i\hbar}{2}\sigma(J_0F,G)\Big(\Phi_\hbar(J_0F)\Phi_\hbar(G) + \Phi_\hbar(G)\Phi_\hbar(J_0F)\Big)\\
            &\ \  +\frac{i\hbar}{2}\sigma(F,G)\Big(\Phi_\hbar(J_0F)\Phi_\hbar(J_0G) + \Phi_\hbar(J_0G)\Phi_\hbar(J_0F)\Big).
            \end{split}
        \end{split}
    \end{equation*}
    Then (4) of Prop. \ref{prop:lims} implies the result.\qedhere
\end{enumerate}
\end{proof}

These results are somewhat restricted.  For example, it is difficult to establish an analogue of von Neumann's condition for number operators because one encounters unbounded operators in the relevant differences.  Still, we take the foregoing to establish some sense in which $a_0^{J_0}(F)$, $(a_0^{J_0}(F))^*$, and $N_0^{J_0}(F)$ are classical limits of $a_\hbar^{J_0}(F)$, $(a_\hbar^{J_0}(F))^*$, and $N_\hbar^{J_0}(F)$.  We recognize, however, that it would be interesting to be able to strengthen the approximations involved in the classical limit for unbounded quantities.

\section{Classical Limit for the Klein-Gordon Field}
\label{sec:KG}

We now analyze the classical limits of number operators and Hamiltonians in the model of a real scalar field $\varphi$ on Minkowski spacetime $\mathbb{R}^4$ satisfying the Klein-Gordon equation
\[\Big(\frac{\partial^2}{\partial t^2} - \nabla^2\Big)\varphi = - m^2\varphi,\]
where $\nabla^2$ is the spatial Laplacian and $m> 0$.  We work with initial data on $\mathbb{R}^3$, defining $E:= C_c^\infty(\mathbb{R}^3)\oplus C_c^\infty(\mathbb{R}^3)$ as the space of pairs of test functions with the symplectic form
\[\sigma((f_1,g_1),(f_2,g_2)) : = \int_{\mathbb{R}^3} f_1g_2 - f_2g_1\]
for all $f_1,f_2,g_1,g_2\in C_c^\infty(\mathbb{R}^3)$.  The phase space $E'$ will be a topological dual to $E$ in some vector space topology such that $C^\infty(\mathbb{R}^3)\oplus C^\infty(\mathbb{R}^3)\subseteq E'$.  The space $E$ consists in pairs $(\pi,\varphi)$ of (possibly distributional) field configurations $\varphi$ and conjugate momenta $\pi:=\frac{\partial\varphi}{\partial t}$.  We will analyze the classical limits of two classes of number operators associated with the scalar field: ``Minkowski" number operators associated with an inertial observer and ``Rindler" number operators associated with an accelerating observer on the right Rindler wedge.

\subsection{Minkowski Number Operators}

To define the Minkowski number operators, we must specify a choice of complex structure.  To do so, we define an operator $\mu_M: C_c^\infty(\mathbb{R}^3)\to C_c^\infty(\mathbb{R}^3)$ by
\[\mu_M:= \big(m^2-\nabla^2\big)^{1/2}.\]
This operator $\mu_M$ is self-adjoint and a bijection (This follows, e.g., from Thm. IX.27, p. 54 of \citeauthor{ReSi75}\cite{ReSi75}).  We define a complex structure $J_M: E\to E$ by
\[J_M(f,g) := (-\mu_M^{-1}g,\mu_M f)\]
for all $f,g\in C_c^\infty(\mathbb{R}^3)$.  $J_M$ is the unique complex structure compatible with time translations with respect to the inertial timelike symmetries of Minkowski spacetime; see \citeauthor{Ka79}.\citep{Ka79,Ka85}  We define the Minkowski number operators $N^M_\hbar(F)$ as the number operators corresponding to this choice of complex structure, i.e., $N^M_\hbar(F):= N^{J_M}_\hbar(F)$ for each $F\in E$.  We use similar notation for $a^M_\hbar(F)$. Explicitly, we have
\begin{align*}
a^M_\hbar(F) &:= \Phi_\hbar(F) + i\Phi_\hbar(J_MF)\\
N_\hbar^M(F) &:= (a^M_\hbar(F))^*a^M_\hbar(F).
\end{align*}

\noindent These number operators are the usual ones defined in the Fock space representation of the Weyl algebra, when the inertial timelike symmetries of Minkowski spacetime are used in the frequency splitting procedure for ``second quantization". \citep{Wa94}

The results of the previous section establish a sense in which $N^M_0(F)$ is the classical limit of $N_\hbar^M(F)$.  We now analyze the contents of $N^M_0(F)$ in the classical field theory.  

We can use this setup to analyze $N_0^M(F)$ as a function on $E'$.  Recall that $\mathcal{W}(E,0)$ is *-isomorphic to the algebra $AP(E')$ of $\sigma(E',E)$-continuous almost periodic functions on $E'$.  In this setting, given test functions $(f,g)\in E$, the classical Weyl unitaries and fields have the form
\begin{align*}
    W_0(f,g)(\pi,\varphi) &= \exp\bigg\{i\int_{\mathbb{R}^3}(\pi f + \varphi g)\bigg\}\\
    \Phi_0(f,g)(\pi,\varphi)&= \int_{\mathbb{R}^3}(\pi f + \varphi g)
\end{align*}
for all field configurations and conjugate momenta $(\pi,\varphi)\in C^\infty(\mathbb{R}^3)\oplus C^\infty(\mathbb{R}^3)\subseteq E'$.   This also immediately determines the form of $N_0^M(f,g)$.

\begin{prop}
For any $(f,g)\in E$,
\[N_0^M(f,g)(\pi,\varphi) = \frac{1}{2}\bigg(\int_{\mathbb{R}^3}\pi f + \varphi g\bigg)^2 + \frac{1}{2}\bigg(\int_{\mathbb{R}^3}\varphi (\mu_M f) - \pi (\mu_M^{-1}g)\bigg)^2\]
for all $(\pi,\varphi)\in C^\infty(\mathbb{R}^3)\oplus C^\infty(\mathbb{R}^3)\subseteq E'$.
\end{prop}

Furthermore, we can construct the \emph{classical Minkowski total number operator} $\overline{N}_0^M$ by letting $\{F_k\}$ be any $\alpha_{J_M}$-orthonormal basis for $E$ and defining
\begin{align*}
    \overline{N}_0^M:=\sum_k N_0^M(F_k).
\end{align*}
The following proposition provides an explicit form for the total number operator as a real-valued function on the phase space.  This establishes that the definition of the total number operator is independent of the chosen basis, which holds similarly for all total number operators and total Hamiltonians in the following sections.

\begin{prop}
\label{prop:Mnumber}
\begin{align*}
    \overline{N}_0^M(\pi,\varphi) = \frac{1}{2}\int_{\mathbb{R}^3}\pi(\mu_M^{-1}\pi) + \varphi(\mu_M\varphi)
\end{align*}
for any $(\pi,\varphi)\in C_c^\infty(\mathbb{R}^3)\oplus C_c^\infty(\mathbb{R}^3)\subseteq E'$.
\end{prop}

\begin{proof}
The Pythagorean theorem for the Hilbert space completion of $E$ with inner product $\alpha_{J_M}$ implies that
\begin{align*}
    \overline{N}_0^M(\pi,\varphi) &= \sum_k N_0^M(F_k)(\pi,\varphi)\\ 
    &= \frac{1}{2}\sum_k\abs{\alpha_{J_M}\big((\varphi,-\pi),F_k\big)}^2\\
    &=\frac{1}{2}\alpha_{J_M}\big((\varphi,-\pi),(\varphi,-\pi)\big)\\
    &=\frac{1}{2}\int_{\mathbb{R}^3}\pi(\mu_M^{-1}\pi) + \varphi(\mu_M\varphi).  \qedhere
\end{align*}
\end{proof}

Finally, we can construct the \emph{classical Minkowski Hamiltonian} $H_0^M$.  Let $\{f_k\}$ be any orthonormal basis for $L^2(\mathbb{R}^3,\mathbb{R})$ and define
\[H^M_0: = \sum_k N_0^M(f_k,0) = \sum_k N_0^M(0,\mu_M f_k)\]
Notice that we take the sum over only one test function component and that we use an orthonormal basis with respect to the $L^2$ inner product rather than the inner product $\alpha_{J_M}$.  With this definition, the classical limit of the Minkowski Hamiltonian takes a familiar form as a real-valued function on phase space.

\begin{prop}
\label{prop:Menergy}
\[H_0^M(\pi,\varphi) =\frac{1}{2} \int_{\mathbb{R}^3} \pi^2 + m^2\varphi^2 + (\nabla \varphi)^2\]
for any $(\pi,\varphi)\in C_c^\infty(\mathbb{R}^3)\oplus C_c^\infty(\mathbb{R}^3)\subseteq E'$.
\end{prop}

\begin{proof}
It follows from the Pythagorean theorem for real Hilbert spaces and the self-adjointness of $\mu_M$ that
\begin{align*}
    H_0^M(\pi,\varphi) &= \sum_k \frac{1}{2}\bigg(\int_{\mathbb{R}^3}\pi f_k\bigg)^2 + \bigg(\int_{\mathbb{R}^3}\varphi (\mu_M f_k)\bigg)^2\\
    &= \sum_k\frac{1}{2}\bigg(\int_{\mathbb{R}^3}\pi f_k\bigg)^2 + \bigg(\int_{\mathbb{R}^3} f_k (\mu_m\varphi)\bigg)^2\\
    &= \frac{1}{2}\int_{\mathbb{R}^3}\pi^2 + (\mu_M\varphi)^2\\
    &= \frac{1}{2}\int_{\mathbb{R}^3}\pi^2 + \varphi(\mu_M^2\varphi)\\
    &= \frac{1}{2}\int_{\mathbb{R}^3}\pi^2 + m^2\varphi^2 - \varphi(\nabla^2\varphi)\\
    &= \frac{1}{2}\int_{\mathbb{R}^3}\pi^2 + m^2\varphi^2 + (\nabla\varphi)^2.
\end{align*}
The final line is implied by the divergence theorem because
\begin{align*}
    \text{div}(\varphi\nabla\varphi) &= \varphi\nabla^2\varphi + (\nabla\varphi)^2. \qedhere
\end{align*}
\end{proof}

This shows that the classical limit of the Minkowski total number operator is equal to the classical total energy of the Klein-Gordon field for initial data of compact support.  This, of course, is the conserved quantity of the Klein-Gordon field corresponding to the inertial timelike symmetries of Minkowski spacetime.

\subsection{Rindler Number Operators}

To define the Rindler number operators, we specify a different choice of complex structure.  We work on the right Rindler wedge $\{(t,x,y,z)\in \mathbb{R}^4\ |\ x > \abs{t}\}$, and so we restrict attention to initial data with support in $R:= \{(x,y,z)\in\mathbb{R}^3\ |\ x>0\}$, and we restrict attention to pairs of test functions in $\overset{\triangleleft}{E}:= C_c^\infty(R)\oplus C_c^\infty(R)$.  For comparison with \citet{Ka85}, we work with functions of the form $e^x f$, and hence we identify each $f\in C^\infty(R)$ with $\overset{\triangleleft}f:= e^x f$.  We proceed as in the previous section by first defining an operator $\mu_R: C_c^\infty(R)\to C_c^\infty(R)$ by
\[\mu_R:= \bigg(e^{2x}\Big(m^2-\frac{\partial}{\partial y^2} - \frac{\partial}{\partial z^2}\Big) - \frac{\partial}{\partial x^2}\bigg)^{1/2}.\]
\citet[][p. 72]{Ka85} establishes that $\mu_R$ is positive and essentially self-adjoint on $\overset{\triangleleft}{E}$.  As in the previous section, we define a complex structure $J_R:\overset{\triangleleft}{E}\to \overset{\triangleleft}{E}$ by
\[J_R(f,\overset{\triangleleft}{g}):=(-\mu_R^{-1}\overset{\triangleleft}{g},\mu_R f)\]
for all $f,\overset{\triangleleft}{g}\in C_c^\infty(R)$.  $J_R$ is the unique complex structure compatible with time translatiosn with respect to the Lorentz boost timelike symmetries of Minkowski spacetime; see \citeauthor{Ka79}.\citep{Ka79,Ka85}  We define the Rindler number operators as the number operators corresponding to this choice of complex structure, i.e., $N_\hbar^R(f,\overset{\triangleleft}{g}):= N_\hbar^{J_R}(f,\overset{\triangleleft}{g})$ for each $(f,\overset{\triangleleft}{g})\in \overset{\triangleleft}{E}$.  We use a similar notation for $a_\hbar^R(f,\overset{\triangleleft}{g})$.  Explicitly, we have
\begin{align*}
    a_\hbar^R(F) &:= \Phi_\hbar(F) + i\Phi_\hbar(J_RF)\\
    N_\hbar^R(F) &:= (a_\hbar^R(F))^*a_\hbar^R(F).
\end{align*}
These number operators correspond to those in the Fock space determined by the one-particle structure for an observer in uniform acceleration, for whom Rindler coordinates on the right Rindler wedge form a natural reference frame for dynamics as described in \S4 of  \citet{Ka85} (see also \citet{KaWa91}).

The results of the previous section establish a sense in which $N_0^R(F)$ is the classical limit of $N_\hbar^R(F)$.  We now analyze the contents of $N^R_0(F)$ in the classical field theory.  That is, we again analyze $N_0^R(F)$ as a function on $\overset{\triangleleft}{E'}$ in the representation of $\mathcal{W}(\overset{\triangleleft}{E},0)$ as $AP(\overset{\triangleleft}{E'})$.  As before, the representation immediately determines the form of $N_0^R(f,\overset{\triangleleft}{g})$.

\begin{prop}
For any $(f,\overset{\triangleleft}{g})\in \overset{\triangleleft}{E}$,
\[N_0^R(f,\overset{\triangleleft}{g})(\overset{\triangleleft}{\pi},\varphi) = \frac{1}{2}\bigg(\int_{R}\overset{\triangleleft}{\pi} f + \varphi \overset{\triangleleft}{g}\bigg)^2 + \frac{1}{2}\bigg(\int_{R}\varphi (\mu_R f) - \overset{\triangleleft}{\pi} (\mu_R^{-1}\overset{\triangleleft}{g})\bigg)^2\]
for all $(\overset{\triangleleft}{\pi},\varphi)\in C^\infty(R)\oplus C^\infty(R)\subseteq \overset{\triangleleft}{E'}$.
\end{prop}

Furthermore, we can construct the \emph{classical Rindler total number operator} $\overline{N}_0^R$ by letting $\{F_k\}$ be any $\alpha_{J_R}$-orthonormal basis for $\overset{\triangleleft}{E}$ and defining
\begin{align*}
    \overline{N}_0^R := \sum_k N_0^R(F_k).
\end{align*}
The following proposition provides an explicit form for the Rindler total number operator as a real-valued function on phase space.

\begin{prop}
\begin{align*}
\overline{N}_0^R(\overset{\triangleleft}{\pi},\varphi) = \frac{1}{2} \int_{R} \overset{\triangleleft}{\pi}(\mu_R^{-1}\overset{\triangleleft}{\pi}) + \varphi(\mu_R \varphi)
\end{align*}
for all $(\overset{\triangleleft}{\pi},\varphi)\in C_c^\infty(R)\oplus C_c^\infty(R)\subseteq \overset{\triangleleft}{E'}$.
\end{prop}

\begin{proof}
This follows from an analogous calculation to that in the proof of Prop. \ref{prop:Mnumber}.
\end{proof}

Finally, we can construct the \emph{classical Rindler Hamiltonian} $H_0^R$.  Let $\{f_k\}$ be any orthonormal basis for $L^2(R,\mathbb{R})$ and define
\[H_0^R:= \sum_k N_0^R(f_k,0) = \sum_k N_0^R(0,\mu_R f_k).\]
With this definition, we have the following explicit form of $H_0^R$ as a real-valued function on phase space.

\begin{prop}
\[H_0^R(\overset{\triangleleft}{\pi},\varphi) = \frac{1}{2}\int_R(\overset{\triangleleft}{\pi})^2 + \varphi(\mu_R^2\varphi)\]
for all $(\overset{\triangleleft}{\pi},\varphi)\in C_c^\infty(R)\oplus C_c^\infty(R)\subseteq \overset{\triangleleft}{E'}$.
\end{prop}

\begin{proof}
This follows from an analogous calculation to that in the proof of Prop. \ref{prop:Menergy} using the self-adjointness of $\mu_R$.
\end{proof}

This expression is the Rindler energy, which is the conserved quantity of the Klein-Gordon field associated with the timelike Lorentz boost symmetries of $R$, i.e., time translations in Rindler coordinates.\citep{Ka85}  Thus, the previous proposition shows that the classical limit of the Rindler total number operator is the Rindler energy.

\section{Classical Limit for the Maxwell field}
\label{sec:em}

In this section, we analyze the classical limit of the Minkowski number operator and Hamiltonian for an electromagnetic field on Minkowski spacetime satisfying the source-free Maxwell equations.  As before, we work with initial data on a surface $\mathbb{R}^3$, on which we assume the electromagnetic field to be decomposed into an electric field vector $E$ with components $E^k$ (for $k=1,2,3$) in some fixed coordinate system and a magnetic (co)vector potential $A$ with components $A_j$ (for $j=1,2,3$) in the Coulomb gauge (satisfying $\text{div} (A) = 0$).  In this formulation, the source-free Maxwell equations take the form
\begin{align*}
\text{div}(E) &= 0\\
\Big(\frac{\partial^2}{\partial t^2} - \nabla^2\Big)A_j &= 0 \ \ \text{(for  $j=1,2,3$)}.
\end{align*}
Thus, each component $A_j$ satisfies the mass zero Klein-Gordon equation.  We take the test function space to be 
\[V:= \{(f,g)\in T^{0,1}_c(\mathbb{R}^3) \oplus T^{1,0}_c(\mathbb{R}^3)\ |\ \text{div}(f) = \text{div}(g) = 0 \}\]
where $T^{0,1}_c(\mathbb{R}^3)$ is the space of smooth, compactly supported covector fields $f$, and $T^{1,0}_c(\mathbb{R}^3)$ is the space of smooth, compactly supported vector fields $g$.  We define the symplectic form on $V$ by
\begin{align*}
    \sigma(({f},{g}),(\tilde{f},\tilde{g})):=\int_{\mathbb{R}^3} \sum_k{f_k}\tilde{g}^k - \tilde{f}_k{g^k}.
\end{align*}
The phase space $V'$ will be the topological dual to $V$ in some vector space topology, consisting of pairs $(E,A)$ of (possibly distributional) field configurations $A$ and conjugate momenta $E^k := -\delta^{kj}(\frac{\partial A}{\partial t})_j$ (the Euclidean metric tensor $\delta^{kj}$ is defined by $\delta^{kj} = 1$ if $k=j$ and 0 otherwise).

We will need the following lemma, which establishes that when $E$ and $A$ are smooth field configurations, they can always be chosen to be divergence free.

\begin{lemma}
Suppose $\rho\in V'$ has the form
\[\rho(f,g) = \int_{\mathbb{R}^3}\sum_k f_kE^k + g^kA_k\]
for every $(f,g)\in V$ for some smooth fields $(E,A)\in T_c^{1,0}(\mathbb{R}^3)\oplus T_c^{0,1}(\mathbb{R}^3)\subseteq V'$.  Then there is a unique pair $(\hat{E},\hat{A})\in T_c^{1,0}(\mathbb{R}^3)\oplus T_c^{0,1}(\mathbb{R}^3)\subseteq V'$ with  $\text{\emph{div}}(\hat{E}) = \text{\emph{div}}(\hat{A}) = 0$ such that
\begin{align*}
    \rho(f,g) = \int_{\mathbb{R}^3}\sum_kf_k\hat{E}^k+g^k\hat{A}_k.
\end{align*}
for all $(f,g)\in V$.
\end{lemma}

\begin{proof}
The fundamental theorem of vector calculus implies that $E$ and $A$ can be decomposed uniquely into curl-free and divergence-free fields
\begin{align*}
    E^k &= \hat{E}^k + \delta^{kj}\nabla_j\phi_E\\
    A_k &= \hat{A}_k + \nabla_k\phi_A
\end{align*}
where $\text{div}(\hat{E}) = \text{div}(\hat{A}) = 0$ and $\phi_E,\phi_A$ are smooth scalar fields.

Now we have that for every $(f,g)\in V$,
\begin{align*}
    \rho(f,g) &= \int_{\mathbb{R}^3}\sum_k f_k E^k + g^kA_k\\
    &= \int_{\mathbb{R}^3}\sum_k f_k\hat{E}^k + f_k\delta^{kj}\nabla_j\phi_E + g^k\hat{A}_k + g^k\nabla_k\phi_A\\
    &= \int_{\mathbb{R}^3}\sum_k (f_k\hat{E}^k + g^k\hat{A}_k) - \phi_E\text{div}(f) - \phi_A\text{div}(g)\\
    &= \int_{\mathbb{R}^3} \sum_k f_k\hat{E}^k + g^k\hat{A}_k,
\end{align*}
as desired.  The second to last line is implied by the divergence theorem, while the last line is implied by the assumption that $f$ and $g$ are divergence free.
\end{proof}

In what follows, when we have linear functionals on $V$ determined by smooth, compactly supported field configurations $(E,A)\in T^{1,0}_c(\mathbb{R}^3)\oplus T_c^{0,1}(\mathbb{R}^3)\subseteq V'$, we will always choose $E$ and $A$ to be divergence free without further comment, as justified by the preceding lemma.  Notice this means that we impose both the Coulomb gauge and the first of the Maxwell equations through the kinematical structure of $V$.  The remaining Maxwell equations are encoded in the choice of complex structure.

To specify a complex structure, we define the operator $\mu_{EM}: T^{0,1}_c(\mathbb{R}^3)\to T^{1,0}_c(\mathbb{R}^3)$ by
\[(\mu_{EM}f)^j:= \delta^{jk}(-\nabla^2)^{1/2}f_k\]
for all $f\in T^{0,1}_c(\mathbb{R}^3)$.  Then we define the complex structure $J_{EM}: V\to V$ by
\[J_{EM}(f,g) = (-\mu_{EM}^{-1}g,\mu_{EM}f)\]
for all $(f,g)\in V$.  Again, $J_{EM}$ is the unique complex structure compatible with inertial time translations of the fields $(E,A)$ satisfying Maxwell's equations.  As in the previous sections, we define the electromagnetic number operators $N_\hbar^{EM}(F)$ as the number operators corresponding to this choice of complex structure, i.e., $N_\hbar^{EM}(F) := N_\hbar^{J_{EM}}(F)$ for each $F\in V$.

We again consider the representation of $\mathcal{W}(V,0)$ as $AP(V')$.  In this representation, we have the following form for $N_0^{EM}(f,g)$ as a real-valued function on phase space.

\begin{prop}
For any $(f,g)\in V$,
\[N_0^{EM}(f,g)(E,A) = \frac{1}{2}\bigg(\int_{\mathbb{R}^3}\sum_k E^kf_k + A_kg^k\bigg)^2 + \frac{1}{2}\bigg(\int_{\mathbb{R}^3}\sum_kA_k(\mu_{EM}f)^k - E^k(\mu_{EM}g)_k\bigg)^2\]
for all $(E,A)\in T^{1,0}(\mathbb{R}^3)\oplus T^{0,1}(\mathbb{R}^3)\subseteq V'$.
\end{prop}

Furthermore, we can construct the \emph{classical total electromagnetic number operator} $\overline{N}_0^{EM}$ by letting $\{F_k\}$ be any $\alpha_{J_{EM}}$-orthonormal basis for $V$ and defining
\[\overline{N}_0^{EM} := \sum_k N_0^{EM}(F_k).\]
The following proposition provides an explicit form for the total number operator.

\begin{prop}
\[\overline{N}_0^{EM}(E,A) = \frac{1}{2}\int_{\mathbb{R}^3} \sum_k E^k(\mu_{EM}^{-1}E)_k + A_k(\mu_{EM}A)^k\]
for any $(E,A)\in T_c^{1,0}(\mathbb{R}^3)\oplus T_c^{0,1}(\mathbb{R}^3)\subseteq V'$.
\end{prop}

\begin{proof}
This follows from an analogous calculation to that in the proof of Prop. \ref{prop:Mnumber}.
\end{proof}

Finally, we can construct the \emph{classical electromagnetic Hamiltonian} $H_0^{EM}$.  Let $\{\overset{k}{f}\}$ be any orthonormal basis for $T^{1,0}_c(\mathbb{R}^3)$ with the generalized $L^2$-inner product
\[\inner{f}{g} = \int_{\mathbb{R}^3} \sum_{l,m} \delta_{jl}{f^j}{g^l}.\]
Now we define
\[H_0^{EM}:=\sum_k N_0^{EM}(\overset{k}{f},0).\]
With this definition, the classical limit of the electromagnetic Hamiltonian also takes a familiar form as a real-valued function on phase space.

\begin{prop}
\[H_0^{EM}(E,A) = \frac{1}{2}\int_{\mathbb{R}^3}\sum_{j,k}\delta_{jk}E^jE^k + \delta^{jk}\text{\emph{curl}}(A)_j \text{\emph{curl}}(A)_k\]
for any $(E,A)\in T_c^{1,0}(\mathbb{R}^3)\oplus T_c^{0,1}(\mathbb{R}^3)\subseteq V'$.
\end{prop}

\begin{proof}
An analogous calculation to that in the proof of Prop. \ref{prop:Menergy} yields

\[H_0^{EM}(E,A) = \frac{1}{2}\int_{\mathbb{R}^3}\sum_{j,k} \delta_{jk} E^jE^k + \delta^{jk}(\nabla^2 A_j)A_k,\]
so it suffices to show that $\int\sum_{j,k}\delta^{jk}(-\nabla^2 A_j)A_k = \int\sum_{j,k} \delta^{jk}\text{curl}(A)_j\text{curl}(A)_k$.

To this end, first note that since we choose $A$ to be divergence free, we have $0 = \partial_1A_1 + \partial_2A_2 + \partial_3A_3$ and hence,
\begin{align*}
    -A_1\partial_1^2A_1 &= A_1\partial_1(\partial_2A_2 + \partial_3A_3)\\
     -A_2\partial_2^2A_2 &= A_2\partial_2(\partial_1A_1 + \partial_3A_3)\\
      -A_3\partial_3^2A_3 &= A_3\partial_3(\partial_1A_1 + \partial_2A_2).
\end{align*}

From this, it follows that
\begin{align*}
    \int_{\mathbb{R}^3}\sum_{j,k}\delta^{jk}(-\nabla^2A_j)A_k &= \int_{\mathbb{R}^3}-A_1(\partial_1^2A_1 + \partial_2^2A_1 + \partial_3^2A_1) -A_2(\partial_1^2A_2 + \partial_2^2A_2 + \partial_3^2A_2)\\
    &\indent\indent - A_3(\partial_1^2A_3 + \partial_2^2A_3 + \partial_3^2A_3)
    \\
    &= \int_{\mathbb{R}^3} (\partial_1A_2)^2 + (\partial_2A_1)^2 + (\partial_2A_3)^2 + (\partial_3A_2)^2 + (\partial_1A_3)^2 + (\partial_3A_1)^2\\
    &\indent\indent -\int_{\mathbb{R}^3}A_1\partial_1^2A_1 + A_2\partial_2^2A_2 + A_3\partial_3^2A_3\\
    &=\int_{\mathbb{R}^3} (\partial_1A_2)^2 + (\partial_2A_1)^2 + (\partial_2A_3)^2 + (\partial_3A_2)^2 + (\partial_1A_3)^2 + (\partial_3A_1)^2\\
    &\indent\indent+\int_{\mathbb{R}^3}A_1\partial_1(\partial_2A_2 + \partial_3A_3)+A_2\partial_2(\partial_1A_1 + \partial_3A_3)+A_3\partial_3(\partial_1A_1 + \partial_2A_2)\\
    &=\int_{\mathbb{R}^3} (\partial_1A_2)^2 + (\partial_2A_1)^2 + (\partial_2A_3)^2 + (\partial_3A_2)^2 + (\partial_1A_3)^2 + (\partial_3A_1)^2\\
    &\indent\indent-2\int_{\mathbb{R}^3}(\partial_1A_2)(\partial_2A_1) + (\partial_2A_3)(\partial_3A_2) + (\partial_1A_3)(\partial_3A_1)\\
    &=\int_{\mathbb{R}^3}(\partial_1A_2-\partial_2A_1)^2 + (\partial_2A_3-\partial_3A_2)^2 + (\partial_1A_3-\partial_3A_1)^2\\
    &=\int_{\mathbb{R}^3}\sum_{j,k}\delta^{jk}\text{curl}(A)_j\text{curl}(A)_k,
\end{align*}
where we obtain the second and fourth equalities from integration by parts.  
\end{proof}

This shows that the classical limit of the electromagnetic Hamiltonian is the classical total energy of the electromagnetic field for initial data of compact support.  This, of course, is the conserved quantity of the electromagnetic field corresponding to the inertial timelike symmetries of Minkowski spacetime.

\section{Conclusion}
\label{sec:con}

In this paper, we have analyzed the classical limits of unbounded quantities and illustrated our methods for number operators and Hamiltonians in linear Bosonic quantum field theories.  Our strategy has stayed close to the framework of strict deformation quantization by (i) looking for norm approximations and (ii) treating physical magnitudes as elements of an abstract partial *-algebra rather than focusing on particular Hilbert space representations.  Using developments in the theory of algebras of unbounded operators, we considered continuous extensions of positive quantization maps.  We used these extensions to prove norm approximations in the classical limit for unbounded quantities including field operators, creation and annihilation operators, and number operators.  We then analyzed the classical limits of number operators and associated Hamiltonians for the Klein-Gordon field theory and the Maxwell field theory.  We established that the methods developed in this paper yield a unified approach to the classical limit for both the Minkowski number operators and the Rindler number operators for the Klein-Gordon field, which are unbounded operators that are typically understood to be affiliated with unitarily inequivalent representations of the Weyl algebra.  In both cases, the classical limits of the associated Hamiltonians are the classical conserved energy quantities associated with certain timelike symmetries, as expected.  Similarly, we established (as expected) that the classical limit of the Hamiltonian for the free Maxwell field is the classical energy of the electromagnetic field.  Thus, the methods developed here for taking classical limits of unbounded operators capture the intended use of the classical limit while extending its application beyond C*-algebras in strict quantization.

\section*{Acknowledgments}

Thanks to Adam Caulton, Charles Godfrey, and the audience of the conference ``Foundations of Quantum Field Theory" (Rotman Institute of Philosophy, 2019) for helpful comments and discussion.  BHF acknowledges support from the Royalty Research Fund at the University of Washington during the completion of this work as well as the National Science Foundation under Grant No. 1846560.

\bibliography{bibliography.bib}
\bibliographystyle{mla.bst}

\end{document}